\documentclass[runningheads]{llncs}
\usepackage{algorithm}
\usepackage{algorithmic}
\usepackage[T1]{fontenc}

\def\orcidID#1{\smash{\href{http://orcid.org/#1}{\protect\raisebox{-1.25pt}{\protect\includegraphics{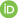}}}}}

\usepackage{pifont}
\usepackage{tikz}
\usetikzlibrary{calc,fit,positioning}
\usetikzlibrary{
    positioning,
    fit,
    arrows.meta,
    backgrounds
}
\usepackage{graphicx}

\usepackage{multicol}
\usepackage{multirow,changepage} 
\usepackage{hhline}
\usepackage{longtable} 
\usepackage{booktabs}
\usepackage{xcolor}

\usepackage{xspace}
\usepackage{todonotes}

\usepackage{amsmath, amssymb}
\usepackage{cite} 
\usepackage{hyperref}

\newcommand{\ZZ}{\mathbb{Z}}

\newcommand{\BB}{\mathbb{B}}
\newcommand{\QQ}{\mathbb{Q}}

\usepackage{siunitx}

\newcommand{\abc}{\textsc{ABC}\xspace}

\newcommand{\amulettwo}{\textsc{AMulet2}\xspace}
\newcommand{\teluma}{\textsc{TeluMA}\xspace}

\newcommand{\dynphaseorderopt}{\textsc{DynPhaseOrderOpt}\xspace}

\newcommand{\talisman}{\textsc{TalisMan1}\xspace}
\newcommand{\talismantwo}{\textsc{TalisMan2.0}\xspace}
\newcommand{\multiling}{\textsc{MultiLinG}\xspace}
\newcommand{\pblib}{\textsc{pblib}\xspace}

\newcommand{\cadical}{\textsc{CaDiCaL}\xspace}
\newcommand{\cmsgen}{\textsc{CMSGen}\xspace}

\newcommand{\flint}{\textsc{FLINT}\xspace}

\def\<#1>{\langle#1\rangle}

\def\specfontfamily#1{\mathcal{#1}}
\def\spec{{\specfontfamily{S}}}

\begin{document}
\title{Avoiding Big Integers: Parallel Multimodular Algebraic Verification of Arithmetic Circuits} %

\titlerunning{Multimodular Circuit Verification}
%
\author{Clemens Hofstadler\inst{1}\orcidID{0000-0002-3025-0604} \and
  Daniela Kaufmann\inst{2}\orcidID{0000-0002-5645-0292} \and
  Chen Chen\inst{3}\orcidID{0009-0007-8825-6260}}
\authorrunning{C. Hofstadler et al.}
%
\institute{Johannes Kepler University, Linz, Austria
  \email{clemens.hofstadler@jku.at}
  \and
  TU Wien, Austria
  \email{daniela.kaufmann@tuwien.ac.at}\and
  Hong Kong University of Science and Technology (Guangzhou), China\\
  \email{cchen099@connect.hkust-gz.edu.cn}}
\maketitle              
\begin{abstract}

Word-level verification of arithmetic circuits with large oper\-ands typically
relies on arbitrary-precision arithmetic, which can lead to significant computational
overhead as word sizes grow. In this paper, we present a hybrid algebraic
verification technique based on polynomial reasoning that combines linear and
nonlinear rewriting. Our approach relies on  multimodular reasoning using homomorphic images, where
computations are performed in parallel modulo different primes, 
thereby avoiding any large-integer arithmetic. We implement the proposed method in the
verification tool \talismantwo and evaluate it on a suite of multiplier
benchmarks. Our results show that hybrid multimodular reasoning significantly
improves upon existing approaches.

\end{abstract}
\section{Introduction}

Word-level reasoning using computer algebra has proven to be an effective
approach for the formal verification of arithmetic circuits, especially for
circuits encoding nonlinear operations, such as multiplier and divider 
circuits~\cite{KonradScholl-FMCAD24,
KaufmannBiereKauers-FMCAD19,KonradSchollMahzoonGrosseDrechsler-FMCAD22,Kaufmann-TACAS25,hofstadler_et_al:LIPIcs.CP.2025.14},
as it avoids the combinatorial blow-up caused by bit-blasting. In this line of
work, circuits are translated into systems of polynomial equations whose
solutions characterize all valid circuit behaviors. Verification then reduces to
proving that a given specification polynomial is implied by the circuit
encoding.

A major limitation of existing word-level algebraic reasoning
approaches is their reliance on large-integer arithmetic. For arithmetic
circuits with word sizes of 64 bits and beyond, the resulting polynomial systems
contain coefficients whose bit-widths exceed native machine word sizes by orders
of magnitude. For instance, the specification of 64-bit multipliers contains
coefficients up to $2^{127}$. As a result, state-of-the-art reasoning tools
typically depend on arbitrary-precision arithmetic, e.g., the implementations
of~\cite{fastpoly,KonradScholl-FMCAD24,KaufmannBiere-TACAS21,talisman-artifact} 
all rely on the GMP library~\cite{gmp}, which incurs computational
overhead and limits scalability in practice.

In this paper, we address this limitation by introducing a verification framework 
based on \emph{parallel multimodular reasoning}. Our approach avoids large-integer
arithmetic entirely by computing \emph{homomorphic images} of the polynomial
system modulo multiple machine-word-sized prime numbers. Rather than reasoning 
over $\mathbb{Z}$ or $\mathbb{Q}$, all computations are carried out
modulo those primes in parallel, ensuring that coefficients remain
within machine word sizes. 
The Chinese remainder theorem (Theorem~\ref{thm:crt}) ensures the correctness,
provided that sufficiently many primes are used (Theorem~\ref{thm:main}).
While homomorphic images are a classical
tool in computer algebra, see e.g.,~\cite[Ch.~5]{gerhard2013modern},
to the best of our knowledge this is the first work to employ them for word-level circuit verification.

To exploit multimodular reasoning effectively, we combine it with a \emph{hybrid
rewriting framework} that integrates both linear and nonlinear algebraic
reasoning. 
Linearization-based techniques~\cite{Kaufmann-TACAS25,
hofstadler_et_al:LIPIcs.CP.2025.14}  extract linear relations from the algebraic encoding of the circuit 
to enable a fast linear rewriting of the previously linearized circuit
specification. 
We build upon the 
guess-and-prove strategy presented in~\cite{hofstadler_et_al:LIPIcs.CP.2025.14}.
The core idea is to sample valid assignments for extracted subcircuits,
from which linear relations are guessed using linear algebra. The
correctness of these guesses is proved using satisfiability
(SAT) solving. In this paper, we significantly improve, parallelize, and adapt this approach to
our multimodular setting. 
In particular, we apply a structure-aware subcircuit selection, 
weighted sampling for more targeted assignments, and dedicated linear algebra routines to better handle a large number of samples. 
All of these optimizations significantly improve the guess-and-prove strategy.

When linear reasoning becomes computationally infeasible due to the
size of the selected subcircuit, our framework switches to nonlinear rewriting, 
using the nonlinear circuit encoding to continue rewriting the specification until completion.  
This combination yields the efficiency of linearization together with the robustness of nonlinear 
reasoning. Additionally, this procedure allows us to always produce counterexamples in case of faulty
circuits, which is not always possible in purely linear approaches.

We implement the proposed approach in our novel verification tool 
\talismantwo~\cite{talismantwo}, 
and evaluate it on a suite of multiplier benchmarks. 
The results demonstrate that parallel multimodular reasoning combining linear and nonlinear rewriting 
is able to efficiently solve a large set of benchmarks while remaining robust.


\section{Preliminaries}

In Section~\ref{sec:prelim-gb}, we review fundamental concepts from polynomial
algebra that are required throughout this work. For a more
comprehensive introduction, we refer to the standard textbooks~\cite{CLO05,CLO15}. In
Section~\ref{ssec:prelim-aig}, we recall And-Inverter Graphs (AIGs) and describe how
they can be represented by polynomial encodings.

\subsection{Algebra}
\label{sec:prelim-gb}

Let $X = \{x_{1},\dots,x_{n}\}$ be a finite set of indeterminates and let
$R$ be a commutative ring with unity.
  A \emph{monomial} (in $X$) is a product of the form
  $x_1^{\alpha_1}\dots x_n^{\alpha_n}$,
  where $\alpha_1,\dots,\alpha_n \in \mathbb N$.
  The set of all monomials in $X$ is denoted by $[X]$.

  A \emph{polynomial} $f$ (over $R$) is a finite
  $R$-linear combination of monomials, i.e.,
  $
    f = c_1 m_{1} + \dots + c_s m_{s},
  $
  with \emph{coefficients} $c_1,\dots,c_s \in R$
  and monomials $m_{1},\dots,m_{s} \in [X]$.
  The product $c_i m_{i}$ is called a \emph{term}.
  The set of all polynomials is denoted by
  $R[X] = R[x_{1},\dots,x_{n}]$.
  
    The \emph{degree} $\deg(m)$ of a monomial
  $m = x_1^{\alpha_1} \dots x_n^{\alpha_n} \in [X]$
  is defined as the sum of its exponents, that is,
  $\deg(m) = \alpha_{1} + \dots + \alpha_{n}$.
  A polynomial $f$ is called \emph{linear} if every monomial $m$ in $f$ has $\deg(m) \leq 1$, otherwise $f$ is called \emph{nonlinear}.

A \emph{monomial order} is a total order $\prec$ on $[X]$ such that $1 \prec m$ for all $m \neq 1$, 
and $m_1 \prec m_2$ implies $m m_1 \prec m m_2$ for all $m,m_1,m_2\in[X]$.  
Typical examples are the \emph{lexicographic order}, which compares exponent vectors componentwise, 
and the \emph{degree lexicographic order}, which first compares degrees and breaks ties lexicographically.  
An order is \emph{degree-compatible} if smaller degree implies smaller order.
With respect to a fixed monomial order, every nonzero
polynomial $f \in R[X]$ has a unique maximal monomial,
called its \emph{leading monomial}. 

A nonempty subset $I \subseteq R[X]$ is an
\emph{ideal} if 
$f + g \in I$ 
and
$hf \in I$ for all $f, g \in I$ and $h \in R[X]$.
A finite set of polynomials
$F = \{f_{1},\dots,f_{r}\} \subseteq R[X]$
can be interpreted as a system of equations
$f_{1} = \dots = f_{r} = 0$.
The set of all polynomials~$h$ whose equation $h = 0$ can be
derived algebraically from this system forms an ideal, denoted
by $\langle F\rangle$. It is given explicitly by
$$
  \langle F\rangle =
  \{ h_{1} f_{1} + \dots + h_{r} f_{r} \mid
  h_{1},\dots,h_{r} \in R[X] \}.
$$

\subsection{And-Inverter Graphs and their Polynomial Encodings}
\label{ssec:prelim-aig}

An \emph{And-Inverter Graph} (AIG)~\cite{Kuehlmann-TCAD2002} is a directed acyclic graph
representing a Boolean circuit $C$. Internal nodes correspond to logical
conjunctions, while edge annotations encode negation. Primary inputs correspond
to Boolean variables.

A \emph{specification} of an AIG is a polynomial $\spec \in R[X]$, where $R$
is a chosen coefficient domain (typically $\mathbb{Z}$) and $X$ contains all input, gate, and output
variables of the circuit. The specification relates the outputs of the circuit
to its primary inputs. Although AIGs operate on Boolean signals, the
specification is not restricted to the Boolean ring $\BB[X]$. 
For instance, the specification of multipliers with inputs
 $a_{n-1}, \ldots, a_0, b_{n-1}, \ldots, b_0$ and outputs $s_{2n-1}, \ldots,
s_0$ can be represented as $\sum_{i=0}^{2n-1}2^{i}s_{i} - \big(
\sum_{i=0}^{n-1}2^{i}a_{i}\big) \cdot \big(\sum_{i=0}^{n-1}2^{i}b_{i} \big)$ over $\ZZ[X]$ or $\QQ[X]$.

To relate the behavior of gates in the circuit to polynomial equations, 
we map $\top \mapsto 1$ and $\bot \mapsto 0$, yielding the
correspondences $x \land y \simeq xy$ and $\neg x \simeq 1-x$ between logic and algebra.
Based on these relations, we define the following \emph{gate polynomials} for gates in a circuit.

\begin{definition}\label{def:gate}
The \emph{gate polynomial} of an AIG node $g$ with inputs $x, y$ is:
  \[
    \arraycolsep=18pt
    \begin{array}{lcr}
      \mbox{Logical constraint} &             & \mbox{Gate polynomial}                \\
      g = x \land y           & \Rightarrow & g - xy                           \\
      g = \neg x \land y        & \Rightarrow & g-(1-x)y \\
      g = x \land \neg y        & \Rightarrow & g-x(1-y) \\
      g = \neg x \land \neg y   & \Rightarrow & g -(1-x)(1-y)\\
      \end{array}
      \]
We denote by $G(C) \subseteq R[X]$ the set of all gate polynomials of a circuit $C$.
\end{definition}

To enforce Boolean semantics when working over coefficient domains $R \neq \mathbb{B}$,
we introduce \emph{Boolean value polynomials}. For each primary
input $x_i$ of the AIG, we define the constraint $x_i^2 - x_i = 0$. We denote the set of
all such polynomials by $B(C)$.
Boolean constraints propagate through the circuit, and it therefore suffices to
introduce Boolean value polynomials only for the primary inputs~\cite{KaufmannBiereKauers-FMSD19}.

\begin{figure*}[tb]
  \centering
  \begin{minipage}{0.59\textwidth}
    
      \scriptsize $
        \begin{array}{l@{\qquad\qquad}l}
          \mbox{Gate Polynomials $G(C)$}                                     & \mbox{Logical constraints}                       \\
           s_3 -                             g_{24}                & s_3  = g_{24}                                \\
           s_2 -                             g_{28}                & s_2  = g_{28}                                \\
           s_1 -                             g_{20}                & s_1  = g_{20}                                \\
           s_0 -                             g_{10}                & s_0  = g_{10}                                \\
          g_{28} - (1-g_{26})(1-g_{24})			 & g_{28} = \neg g_{26} \land \neg g_{24} \\
          g_{26} - (1-g_{22})(1-g_{16})			& g_{26} = \neg g_{22} \land \neg g_{16} \\
           g_{24} - g_{22}g_{16}                             & g_{24} = g_{22} \land g_{16}           \\
           g_{22} - b_1a_1                                         & g_{22} = b_1 \land a_1                       \\
           g_{20} - (1-g_{18})(1-g_{16})			& g_{20} = \neg g_{18} \land \neg g_{16} \\
           g_{18} - (1-g_{14})(1-g_{12})			& g_{18} = \neg g_{14} \land \neg g_{12} \\
           g_{16} - g_{14}g_{12}                             & g_{16} = g_{14} \land g_{12}           \\
           g_{14} - b_1a_0                                         & g_{14} = b_1 \land a_0                       \\
           g_{12} - b_0a_1                                         & g_{12} = b_0 \land a_1                       \\
           g_{10} - b_0a_0                                         & g_{10} = b_0 \land a_0                       \\
        \end{array}
      $ \vspace{2ex} 

  $B(C)$: $a_1^2-a_1, a_0^2-a_0, b_1^2-b_1, b_0^2-b_0$ \vspace{1ex}

  Spec $\spec: 8s_3+4s_2+2s_1+s_0 - 4a_1b_1 - 2a_1b_0-2a_0b_1-a_0b_0$
  \end{minipage}
  \begin{minipage}{0.35\textwidth}
    \includegraphics[width=.99\textwidth]{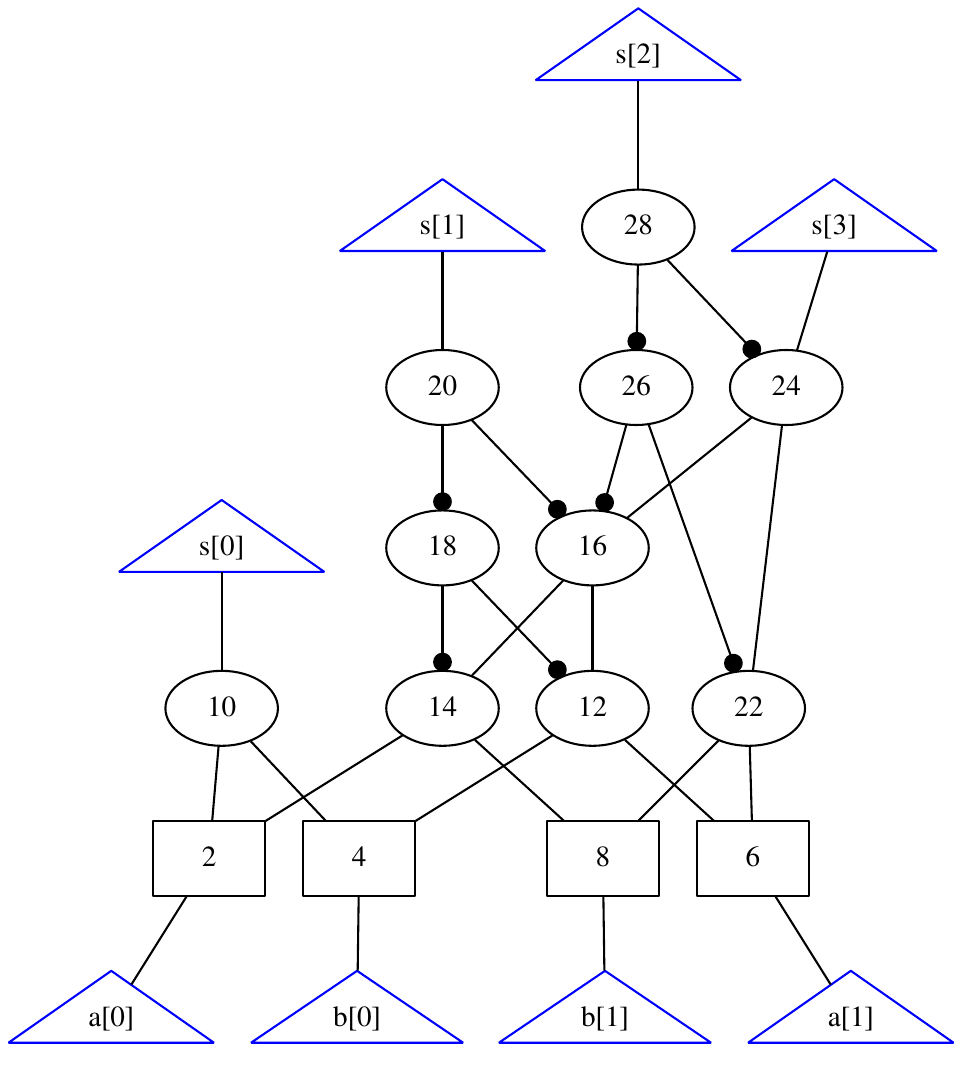}
  \end{minipage}\hfill

  \vspace{2ex}
  
  \caption{AIG and polynomial encoding of a 2-bit multiplier.}
  \label{fig:btor2} 
\end{figure*}

\begin{example}\label{example:gate}
  Figure~\ref{fig:btor2}
   shows an AIG representing a 2-bit
  multiplier and its corresponding polynomial encoding.
  We denote the primary
  inputs by $a_0, a_1, b_0, b_1$ and outputs by
  $s_0, s_1, s_2, s_3$. The internal nodes are denoted by
  $g_i$, where $i$ corresponds to the respective AIG node.
  The specification $\spec$ expresses that the weighted sum of output bits (= output bit-vector) is equal to
  the product of the weighted sum of input bits (= input bit-vectors).
\end{example}

To relate a circuit to its specification, we recall the notion of models from~\cite{KaufmannBiereKauers-FMCAD19}.
First, recall that a \emph{(Boolean) assignment} is a function $\varphi\colon X \to \{0,1\} \subseteq R$ that maps the variables in $X$ to the values 0 and 1.
Any assignment $\varphi$ can be extended to an evaluation of polynomials by substituting $\varphi(x)$ for each variable $x \in X$.
For a polynomial $p \in R[X]$, we denote the evaluation under $\varphi$ by $\varphi(p) \in R$.

\begin{definition}
Let $P \subseteq R[X]$ be a set of polynomials.
A (Boolean) assignment~$\varphi$ is a \emph{model} of $P$ (over $R$) if $\varphi(p) = 0$ for all $p \in P$. 
Moreover, for $q \in R[X]$, we write $P \models_{R} q$ if every model for $P$ is also a model for $\{q\}$,
i.e., whenever $\varphi(p) = 0$ for all $p \in P$, then also $\varphi(q) = 0$.
\end{definition}

We say that a circuit $C$ \emph{fulfills} a specification $\spec$ (over $R$) if $G(C) \models_{R} \spec$.
Note that, since we consider only Boolean assignments, this is equivalent to $G(C) \cup B(C) \models_{R} \spec$.
By~\cite[Thm.~2,3]{KFBK24}, this holds if and only if the specification polynomial $\spec$ lies in the ideal $\langle G(C) \cup B(C) \rangle$.
Hence, correctness of $C$ can be verified algebraically by checking whether $\spec \in \langle G(C) \cup B(C) \rangle$.
Classically, this computation is performed over $R = \ZZ$ (or $R = \QQ$ when a field is required).
In Section~\ref{sec:multimod}, we present a multimodular approach that reduces verification over $\ZZ$ to verification over sufficiently many finite rings $\ZZ_{m_i}$.

\section{Related Work}

The ideal membership 
$
\spec \overset{?}{\in} \langle G(C) \cup B(C) \rangle
$
is decidable using Gröbner basis theory~\cite{Buchberger65}. The set
$G(C) \cup B(C)$ needs to be transformed into a Gröbner basis, which can be of double 
exponential complexity in the number of variables~\cite{mayr1982complexity}. 
However, if a lexicographic monomial order is chosen such that variables follow 
a reverse topological order of
the circuit, the set~$G(C) \cup B(C)$ already forms a Gröbner
basis~\cite{KFBK24}. Thus, correctness of the
circuit can be decided by computing the remainder of~$\spec$ under multivariate
polynomial division with respect to $G(C) \cup B(C)$, see~\cite[Ch.~2.3]{CLO15}. We refer to this approach
as \emph{nonlinear rewriting}.

For multiplier verification, nonlinear rewriting with respect to a
lexicographic order is known to suffer from severe \emph{monomial
blow-up}. Since the leading monomial of gate polynomials typically has lower degree than
the remaining terms (cf.\ Definition~\ref{def:gate}), the degree of intermediate results may
increase substantially. 
This frequently produces intermediate results with millions of monomials~\cite{MahzoonGrosseDrechsler-ICCAD18} and drastically slows down computations.
To mitigate this effect, several advanced rewriting engines have been proposed
in recent years. Early work employed dedicated preprocessing and incremental
reasoning~\cite{KaufmannBiereKauers-FMCAD19}, where syntactically identified parts
of the multiplier are rewritten using SAT solving prior to applying a
column-wise verification algorithm~\cite{KaufmannBiere-TACAS21}. Subsequent
refinements eliminated the reliance on an external SAT solver by introducing a
syntactic algebraic encoding that incorporates literal
polarities~\cite{KaufmannBeameBiereNordstrom-DATE22}. However, the dependence on
syntactic heuristics remains a major limitation of these approaches, as they are often not robust for unseen circuit architectures, such as
those generated by logic synthesis.

Orthogonal to this line of work, dynamic rewriting determines the
rewrite order on the fly and backtracks if intermediate results exceed a
predefined threshold~\cite{MahzoonGrosseSchollDrechsler-DATE20}. This idea was
further extended by incorporating literal polarities into the encoding~\cite{KonradScholl-FMCAD24}. While this approach is more robust across
different circuit architectures, it remains infeasible for certain multiplier
designs.

Recent approaches depart from the lexicographic ordering and
instead employ \emph{degree-compatible} orderings to prevent degree growth
during rewriting~\cite{Kaufmann-TACAS25, hofstadler_et_al:LIPIcs.CP.2025.14}. By
linearizing the specification, rewriting is restricted to linear
polynomials, which avoids the monomial blow-up. 
We refer to this approach as \emph{linear rewriting}.

Linear rewriting comes at the cost of having to extract suitable linear relations from the circuit.
To address this, on-the-fly techniques have been proposed that identify
subcircuits and compute Gröbner bases locally to derive the required linear
relations~\cite{Kaufmann-TACAS25}, as well as approaches based on algebraic
transformations and guess-and-prove techniques~\cite{hofstadler_et_al:LIPIcs.CP.2025.14}.
Among these, the guess-and-prove approach has shown particular promise, as it
can handle comparatively large subcircuits ($>1000$ nodes). In this technique, a subcircuit is
selected and valid assignments for that subcircuit are generated. Linear algebra is then used to
derive candidate linear relations (``guesses''), which subsequently need to be proven, for
instance, using SAT solving. If a candidate relation is invalid, a counterexample
returned by the SAT solver can be used to refine the guess. If no suitable
relation is found, the subcircuit is expanded and the process is repeated, 
see~\cite[Sec.~3.2, 4]{hofstadler_et_al:LIPIcs.CP.2025.14} for a more in-depth discussion.

Despite its effectiveness, the guess-and-prove approach has notable
limitations. First, performance degrades significantly once the extracted
subcircuit exceeds a certain size ($>5000$ nodes). Second, due to large coefficients
in guessed polynomials, the SAT-based proofs can become infeasible in practice, as
the coefficients of the derived polynomials must be encoded at the bit level.

Our work addresses these shortcomings by applying linear rewriting via the guess-and-prove approach using
\emph{multimodular} computations, which impose explicit bounds on coefficients. 
Moreover, whenever linear extraction becomes too costly we propose to switch to nonlinear rewriting, 
thereby combining the strengths of both paradigms in a hybrid approach.


\section{Multimodular Verification via Homomorphic Images}
\label{sec:multimod}

Algebraic word-level verification of arithmetic circuits with real-world input sizes of 64 bits and beyond faces a fundamental scalability bottleneck: 
intermediate coefficients arising during symbolic reasoning may become extremely large. 
As a consequence, verification procedures must rely on arbitrary-precision arithmetic, 
which is significantly slower than native machine-word operations that can exploit fast hardware arithmetic.

Such \emph{expression swell} is a well-known challenge in computer algebra. 
Classic examples include linear system solving and Gröbner basis computations over the integers or rationals.
To address this issue, computer algebra systems often employ \emph{homomorphic images}~\cite{brown1971euclid,arnold2003modular,stein2007modular,hofstadler2025modular}.
The central idea of homomorphic images is to map computations from an unbounded domain, such as $\ZZ$ or $\QQ$,
into finite algebraic structures, typically into $\ZZ_{m}$, the integers modulo $m$, 
where all elements have a fixed-size representation and no expression swell can occur. 
Typically, $m$ is chosen to be a prime number, such that $\ZZ_{m}$ is a finite field, 
ensuring desirable algebraic properties for reasoning and computation.

Concretely, we translate the idea into the following multimodular verification strategy:
Instead of verifying that a circuit $C$ fulfills a specification $\spec$ over $\ZZ$,
we perform the verification modulo several integers $m_{1},\dots,m_{k}$ (typically primes).
Working in $\ZZ_{m_{i}}$ ensures that coefficients remain bounded, 
and when the moduli $m_{i}$ are chosen to fit within machine words, we can rely on fast native arithmetic.
The Chinese remainder theorem ensures the correctness of the circuit over the integers,
provided that sufficiently many moduli are used.
We recall the Chinese remainder theorem below and refer to, e.g.,~\cite[Thm.~4.12]{crt} for a proof.

\begin{theorem}[Chinese Remainder Theorem]\label{thm:crt}
Let $m_1,\dots,m_k \in \mathbb{N}$ be pairwise coprime positive integers.
For any $a_{1},\dots,a_{k} \in \ZZ$, the system of linear congruences
$x \equiv a_{i} \bmod m_{i}$ for $i = 1,\dots,k$
has a unique solution modulo $\prod_{i=1}^{k} m_{i}$.
\end{theorem}

To formally justify our approach, we first observe that the set of models of a circuit does not depend on whether computations are carried 
out over the integers or modulo an integer $m \geq 2$ before proving our main result in Theorem~\ref{thm:main}.

\begin{lemma}
\label{lem:models}
Let $C$ be an AIG and $m \geq 2$.
A Boolean assignment $\varphi$ is a model of $G(C)$ over $\mathbb{Z}$ 
if and only if it is a model of $G(C)$ over $\mathbb{Z}_m$.
\end{lemma}

\begin{proof}
If $\varphi(p) = 0$ over $\mathbb{Z}$, then clearly $\varphi(p) \equiv 0 \bmod m$.
Thus, any model over $\mathbb{Z}$ is also a model over $\mathbb{Z}_m$.
Conversely, assume that $\varphi$ is not a model of $G(C)$ over $\mathbb{Z}$, i.e., there exists $p \in G(C)$ such that $\varphi(p) \neq 0$.
Since $p$ is a gate polynomial, $\varphi(p) \in \{-1,1\}$.
For $m \geq 2$, both values remain nonzero modulo $m$, implying that $\varphi(p) \not\equiv 0 \bmod m$.
So, $\varphi$ is also not a model of $G(C)$ over $\mathbb{Z}_m$.
\end{proof}


\begin{theorem}
\label{thm:main}
Let $C$ be an AIG and let $\spec \in \ZZ[X]$ be a specification of $C$.
Further, let $m_1,\dots,m_k \in \mathbb{N} \setminus \{0,1\}$ be pairwise coprime positive integers such that $\lvert \varphi(\spec)\rvert < \prod_{i=1}^k m_i$ for all Boolean assignments $\varphi$.
Then
\begin{align*}
    G(C) \models_\mathbb{Z} \spec \qquad \text{iff} \qquad
    \forall i \in \{1,\dots,k\} : G(C) \models_{\mathbb{Z}_{m_i}} \spec.
\end{align*}
\end{theorem}
\begin{proof}

First, assume that $G(C) \models_\mathbb{Z} \spec$ holds.
Let $i \in \{1,\dots,k\}$ be arbitrary and let $\varphi$ be a model of $G(C)$ over $\mathbb{Z}_{m_i}$.
By Lemma~\ref{lem:models}, $\varphi$ is also a model of $G(C)$ over $\mathbb{Z}$.
Thus, by assumption $\varphi(\spec) = 0$ over $\mathbb{Z}$, which implies $\varphi(\spec) \equiv 0 \bmod m_i$, 
showing that $\varphi$ is also a model of $\spec$ over $\mathbb{Z}_{m_i}$.

Conversely, assume that $G(C) \models_{\mathbb{Z}_{m_i}} \spec$ holds for all $i = 1,\dots,k$.
Let $\varphi$ be a model of $G(C)$ over $\mathbb{Z}$.
Again by Lemma~\ref{lem:models}, $\varphi$ is also a model of $G(C)$ over each $\mathbb{Z}_{m_i}$.
Hence, $\varphi(\spec) \equiv 0 \bmod m_i$ for all $i$ by assumption.
By the Chinese remainder theorem, this implies 
$\varphi(\spec) \equiv 0 \bmod \prod_{i=1}^k m_i$.
Since $\lvert \varphi(\spec) \rvert < \prod_{i=1}^k m_i$, 
we have $\varphi(\spec) = 0$ over $\mathbb{Z}$, 
showing that $\varphi$ is also a model of $\spec$ over $\mathbb{Z}$.
\end{proof}

Theorem~\ref{thm:main} establishes the correctness of the multimodular verification approach, but it is deliberately agnostic about how the modular verifications
of 
$G(C) \models_{\mathbb{Z}_{m_i}}\spec$ are performed. 
In particular, this could be done by purely algebraic methods, but also using SAT-based reasoning, or any combination thereof.
In the following Section~\ref{sec:hybrid}, we will discuss how to perform the modular verification using our hybrid algebraic reasoning approach.

Note that Theorem~\ref{thm:main} does not require the moduli $m_{i}$ to be prime numbers.
It suffices that they are pairwise coprime.
Nevertheless, it is beneficial to choose the $m_i$ to be primes, so that each~$\mathbb{Z}_{m_i}$ forms a finite field, 
which simplifies subsequent computations that rely on field properties.
By using word-sized primes, all coefficients fit into native machine words, avoiding arbitrary-precision arithmetic altogether.
In particular, as witnessed by implementations of other multimodular algorithms~\cite{monagan2015compact, f4ncgb}, 
it can be beneficial to use 32-bit (resp.~16-bit) prime moduli while storing intermediate results in 64-bit (resp.~32-bit) integers.
This way several arithmetic operations can be accumulated before applying a modular reduction, thereby reducing the number of expensive modulo operations.

The required number of moduli depends on an upper bound for the absolute value of the evaluation of $\spec$ on all Boolean assignments $\varphi$.
Such a bound can be easily obtained directly from $\spec$: 
since we consider only Boolean assignments, each monomial in $\spec$ evaluates to either $0$ or $1$.
Hence, $\lvert\varphi(\spec)\vert$ is bounded by the maximum of the sum of all positive coefficients of
 $\spec$ and the absolute value of the sum of all negative coefficients of $\spec$.
For example, if $\spec$ is the specification of an $n$-bit signed or unsigned multiplier, 
then $\lvert \varphi(\spec) \rvert < 2^{2n}$ for all $\varphi$.

Consequently, for 64-bit multipliers, five 32-bit primes suffice to meet the condition of Theorem~\ref{thm:main}.
Since the verification tasks for different moduli are completely independent, they can be run in parallel, 
yielding near-linear speedups on modern multicore machines.

Finally, we note that the multimodular approach also supports the computation of \emph{witnesses} of faulty circuits.
By Lemma~\ref{lem:models}, any counterexample discovered during verification over a modulus $\mathbb{Z}_{m_i}$, that is, 
any model of $G(C)$ that violates $\spec$ modulo $m_i$, 
also forms a valid counterexample over $\mathbb{Z}$, and thus directly witnesses incorrect circuit behavior.

\section{Hybrid Algebraic Circuit Verification}
\label{sec:hybrid}

In the following, we combine the complementary strengths of linear and nonlinear rewriting
in a hybrid algebraic verification framework that also makes use of the multimodular approach presented in Section~\ref{sec:multimod}.
We also discuss implementation details of this hybrid approach in our tool \talismantwo~\cite{talismantwo} that targets verification of multiplier circuits.

\begin{figure}[tb]
\scalebox{0.70}{
\begin{tikzpicture}[
    node distance=0.7cm and 0.7cm,
    box/.style={
        draw,
        rectangle,
        align=center,
        minimum height=0.85cm,
        text width=1.8cm
    },
    smallbox/.style={box, fill=white, rounded corners=3pt},
]


\node[smallbox] (abc) {Preprocess};
\node[smallbox, right=of abc] (sub) {Extract\\Subcircuit};
\node[smallbox, double, double distance=0.8pt, right=of sub] (sample) {Sample};
\node[smallbox, double, double distance=0.8pt, right=of sample] (guess) {Guess};
\node[smallbox, double, double distance=0.8pt, right=of guess] (prove) {Prove};
\node[smallbox, right=of prove] (reduceL) {Rewrite};
\node[smallbox, double, double distance=0.8pt, below=of prove] (repair) {Repair};

\draw[->,line width=1pt] (abc) -- (sub);
\draw[->,line width=1pt] (sub) -- (sample);
\draw[->,line width=1pt] (sample) -- (guess);
\draw[->,line width=1pt] (guess) -- (prove);
\draw[->,line width=1pt] (prove) -- (reduceL);

\draw[->, bend left=25,line width=1pt] (prove) to (repair);
\draw[->, bend left=25,line width=1pt] (repair) to (prove);

\draw[->,line width=1pt,rounded corners=3pt]
  (reduceL.south)
  -- ++(0,-1.75)
  -- ($(sub.south)+(0,-1.75)$)
  -- (sub.south);

\begin{scope}[on background layer]
\node[
  draw,
  rectangle,
  blue!65,
  fill=blue!10,
  rounded corners=3pt,
  fit=(abc)(sub)(sample)(guess)(prove)(reduceL)(repair),
  inner sep=12pt
] (linear) {};
\end{scope}

\node[
  draw=blue!75,
  fill=blue!75,
  rounded corners=3pt,
  inner sep=4pt,
  anchor=center,
  font=\small\bfseries
] at (linear.north)
{\color{white}Linear Rewriting (Sec.~\ref{ssec:linear})};


\node[smallbox, below=1cm of linear] (simplify) {Simplify};
\node[smallbox, right=of simplify] (reduceN) {Rewrite};

\draw[->,line width=1pt] (simplify) -- (reduceN);

\draw[->,line width=1pt,rounded corners=3pt] (reduceL.south) -- ++(0,-1.75) -- ($(sub.south) + (0,-1.75)$) -- ($(sub.south) + (0,-3.45)$) -- (simplify.west);

\begin{scope}[on background layer]
\node[
    draw,
    rectangle,
    rounded corners=3pt,
    blue!75,
    fill=blue!10,
    fit=(simplify)(reduceN),
    inner sep=12pt
] (nonlinear) {};
\end{scope}

\node[
  draw=blue!75,
  fill=blue!75,
  rounded corners=3pt,
  inner sep=4pt,
  anchor=center,
  font=\small\bfseries
] at (nonlinear.north)
{\color{white}Nonlinear Rewriting (Sec.~\ref{ssec:nonlinear})};


\node[
    box,
    text width=2.4cm,rounded corners=3pt,
    fill=blue!10,
    below=2.6cm of abc
] (enc) {\vspace{4pt}\\ 
Circuit \& Spec\\ encoding};

\node[
  draw=blue!75,
  fill=blue!75,
  rounded corners=3pt,
  inner sep=4pt,
  anchor=center,
  font=\small\bfseries
] at (enc.north)
{\color{white}Input};

\draw[->,line width=1pt]
    ($(enc.north) + (0,0.27)$)
    to
    (abc.south);

\draw[->,line width=1pt] ($(enc.south)-(0,0.5)$) -- (enc.south);

\node[draw,green!60!black,line width=1pt,rounded corners=3pt,
      right=1.2cm of nonlinear,yshift=0.3cm,
      text=green!60!black] (success)
      {Circuit Correct\phantom{I} \ding{51}};

\node[draw,red!60!black,line width=1pt,rounded corners=3pt,
      right=1.2cm of nonlinear,yshift=-0.3cm,
      text=red!60!black] (fail)
      {Circuit Incorrect \ding{55}};

\draw[->,line width=1pt] ([xshift=12pt]reduceL.south) -- ([xshift=14pt]success.north); 
\draw[->,line width=1pt] ($(reduceN.east) + (0,0.3)$) -- (success.west); 
\draw[->,line width=1pt] ($(reduceN.east) + (0,-0.3)$) -- (fail.west);

\end{tikzpicture}
}
\caption{Overview of the \talismantwo workflow.
Double-edged boxes are computations that are run in parallel in \talismantwo.}
\label{fig:flowchart}
\vspace{-3ex}
\end{figure}
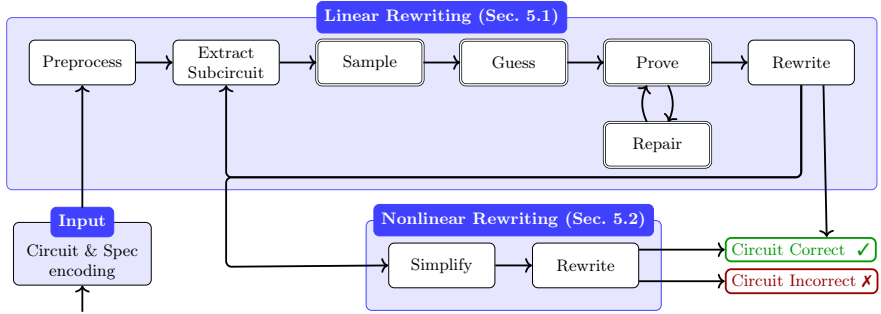 

The overall workflow is shown in Figure~\ref{fig:flowchart}.
Given an input AIG~$C$, let $X$ be the set of all input, output and internal gate variables of $C$. 
We encode the gate constraints $G(C)$ and the Boolean value polynomials $B(C)$ as described in Section~\ref{ssec:prelim-aig}. 
A given specification $\spec \in \mathbb{Z}[X]$ is linearized using extension variables~\cite{Kaufmann-TACAS25,hofstadler_et_al:LIPIcs.CP.2025.14}:
for every nonlinear monomial $m$ in $\spec$, we introduce a new variable $v \notin X$ that acts as a
placeholder. The corresponding constraints $E(C) = \{v-m \mid m
\in \spec \land \deg(m) > 1 \land v \notin X\}$ are added to the circuit
encoding $G(C) \cup B(C) \cup E(C)$. 

Verification then proceeds in two phases. 
We first apply linear rewriting as long as linear relations can be derived efficiently. 
Once this becomes ineffective, we switch to nonlinear rewriting. 
Crucially, both phases are executed over multiple prime fields using homomorphic images, 
ensuring that all arithmetic remains word-sized and that candidate relations can be validated efficiently.
The internal structure of both phases is explained in the following sections.

We note that the approach also works with non-prime moduli, since correctness ultimately relies on Theorem~\ref{thm:main}, which only requires pairwise coprime moduli. 
However, we use primes in practice as they ensure that each residue ring is a field. 
This is particularly useful in the linear rewriting phase, where systems of linear equations are solved to detect candidate relations (see the following section).
While this can also be done over composite moduli (e.g., via Hermite or Smith normal form), Gaussian elimination over prime fields is more efficient.

\subsection{Linear Rewriting} 
\label{ssec:linear}

Linear rewriting aims to simplify the specification polynomial by iteratively
deriving and applying linear relations implied by the circuit. After encoding
and linearizing the specification, we first perform a preprocessing step to
identify recurring arithmetic structures that induce linear constraints. The
resulting relations are collected for subsequent use in rewriting.

The main procedure then iteratively considers the current leading monomial of
the specification polynomial.
To determine the leading monomial we use a reverse topological ordering of the variables.
If a previously derived linear relation is
applicable, it is used to rewrite the monomial. Otherwise, we select a suitable
subcircuit and guess a candidate linear relation suggested by its structure. The
guessed relation is then formally proved correct with respect to the circuit
semantics before it is applied for rewriting.
If no suitable relation can be found efficiently, we fall back to nonlinear rewriting.

\medskip
\noindent
\textbf{Preprocessing.}\;
Many linear relations in arithmetic circuits originate from small building
blocks such as half adders (HA) or full adders (FA). For example, a FA with
inputs $x,y,z$, sum $s$, and carry $c$ satisfies the linear relation $2c+s-x-y-z
= 0$. Our first goal is to identify all half and full adders in the AIG. Instead
of relying on syntactic pattern matching like in previous work~\cite{KaufmannBiereKauers-FMCAD19,KaufmannBeameBiereNordstrom-DATE22}, 
which may be obscured by synthesis, we
detect arithmetic structure semantically using cut enumeration (up to three
leaves) and local truth tables, which we have implemented in our own modified
version of \abc~\cite{abc}. For each detected HA/FA, the corresponding linear
polynomial is constructed and cached for use during linear rewriting.

Beyond local half and full adders, verification based on linear rewriting also
requires us to identify higher-level adder blocks that are used to sum
bit-vectors, such as the final-stage adder (FSA) in multipliers. Since the
linear relations of the FSA generally depend on its global structure, it is
necessary to identify a subcircuit that captures the FSA in its entirety. 
During preprocessing, we use our modified version of \abc to identify
an approximate FSA subcircuit by traversing the circuit backwards 
from the outputs until reaching previously detected HA/FA blocks, which form a 
structural boundary between the FSA and the remaining circuit.
Deriving a topological closure ensures that the resulting
boundary forms a valid cut. Since the input size of the FSA easily exceeds 64
nodes, we do not attempt to derive its linear relations using truth tables as we
did for FA/HA. Instead, they are derived later in the guess-and-prove phase,
where we treat the FSA as a single extracted subcircuit. Hence, this
approximation may safely over-approximate the~FSA.


\medskip
\noindent
\textbf{Extract Subcircuit.}\;
During rewriting, we consider the leading monomial of the current specification. 
If a cached linear polynomial with matching leading monomial exists, we immediately apply it for a rewrite step. 
Otherwise, we attempt to derive such a relation via guess-and-prove on a suitable subcircuit.
If the corresponding gate is part of the FSA approximation, we consider the entire FSA subcircuit.
If not, we extract a depth-bounded subcircuit rooted at that gate. Notably, all nodes whose children are contained in the 
subcircuit will also be added to the subcircuit. Hence, the subcircuit may have more than one output node. 
If no suitable linear relation is found, we increase the depth bound and repeat the procedure with an enlarged subcircuit.
In the worst case we would need to increase the subcircuit until it covers the whole circuit. 

\medskip
\noindent
\textbf{Sample.}\;
To infer candidate relations in a suitable subcircuit, we generate samples of assignments to the subcircuit variables. 
In our previous work~\cite{hofstadler_et_al:LIPIcs.CP.2025.14}, we relied on uniform
sampling over the inputs, with the values of internal gate variables being determined 
by propagation through the circuit structure. 
While this approach is simple and fast, it can miss rare but semantically relevant behaviors. 
This limitation becomes apparent in subcircuits that include long sequences
of AND gates. The output of an $n$-ary AND gate $x = x_0 \land \dots \land x_{n-1}$ is $1$ if and only if all $n$ inputs $x_i$
are set to $1$. 
Under uniform input sampling, this assignment occurs with probability
$2^{-n}$. 
As a result, the guessing procedure most likely observes only assignments for
which $x = 0$, leading to incorrect guesses.

To address this issue, we replace uniform input sampling with the weighted
uniform-like sampler \cmsgen~\cite{weightedsampling}. Unlike uniform sampling
over subcircuit inputs, \cmsgen constructs uniformly sampled assignments over \emph{all}
circuit variables.
It does this incrementally by propagating implications from the current assignment and learning from encountered conflicts.
This significantly improves coverage of rare behaviors as the results in Section~\ref{sec:exp} show.
We typically generate three times as many samples as the subcircuit size.

\medskip
\noindent
\textbf{Guess.}\;
After sampling, the candidate linear relations modulo a prime $p$ are
computed by solving the linear system $Ax\equiv 0$ over $\mathbb{Z}_p$, 
where the columns of the matrix $A$ correspond to variables and the rows to samples,
cf.~\cite[Alg.~2]{hofstadler_et_al:LIPIcs.CP.2025.14}. 

In \talismantwo, we now apply a trick to speed up this linear algebra step. 
The matrices arising in
our context have a special structure: they contain only entries from $\{0,1\}$
and they have a lot more rows than columns. We exploit this structure by solving
the system $(BA)x \equiv 0$ instead of $Ax \equiv 0$, where $B$ is a randomly chosen
$\{0,1\}$-matrix. Since both $A$ and $B$ have entries in $\{0,1\}$, the product
$BA$ can be computed efficiently (over $\mathbb{Z}_{p}$) using
bitwise operations. By choosing $B$ to have fewer rows than $A$, the resulting
system becomes smaller and thus easier to solve. In our implementation, $B$ is
chosen so that $BA$ is a square matrix. The resulting (square) systems are
solved using the number theory library \flint~\cite{flint}. 
Every solution of the original system $Ax \equiv 0$ is also a solution of $(BA)x \equiv 0$,
and if $B$ is chosen randomly, then the two solution spaces coincide with high probability.
Any spurious solutions introduced by the projection would be eliminated during the subsequent prove step.
In our experiments, we never encountered such spurious solutions.

\medskip
\noindent
\textbf{Prove \& Repair.}\;
We prove the correctness of the linear guesses using SAT solving.
To establish that a linear guess $f$ is correct over $\mathbb{Z}_{p}$, we must
show that for all models $\varphi$ of the extracted subcircuit, $\varphi(f) \equiv 0 \bmod p$. 
We translate each gate of the extracted subcircuit into conjunctive normal form (CNF) using Tseitin transformation, 
and encode $f \neq 0$ in CNF using
\pblib~\cite{pblib.sat2015}. We then prove that the formula is unsatisfiable using
the SAT solver~\cadical~\cite{BiereFallerFazekasFleuryFroleyks-CAV24}. 

However, 
encoding $f \neq 0$ alone is actually too weak in our modular setting, since $f$ may evaluate to a non-zero
multiple of $p$. Therefore, whenever the SAT solver returns a satisfying
assignment $\varphi$, we explicitly check whether $\varphi(f) \equiv~0 \bmod p$. If this is
the case, the assignment is excluded from
the search space and the SAT call is repeated.
We have experimented with directly  encoding $f \neq 0 \land f \neq p \land \dots \land f \neq lp$ for a suitable upper bound $l$. 
However, this caused a significant overhead, even for $l = 2$.
Additionally, we also experimented with mixed-integer programming, but did not observe any significant difference.

Only if we obtain a satisfying assignment with $\varphi(f) \not\equiv 0 \bmod p$ do we accept it as a witness that the guess $f$ is incorrect. 
We add $\varphi$ as a new sample to the matrix $BA$ to repair the linear guesses and re-solve the updated linear system $(BA)x \equiv 0$. 
The resulting refined guesses must then be verified again.

\medskip
\noindent
\textbf{Rewrite.}\;
Once all correct guesses have been collected, we use them to rewrite the specification as far as possible.
In \talismantwo, we perform rewriting steps modulo different primes simultaneously by representing each coefficient
of a multimodular polynomial $f$ as a vector, where the $i$th entry stores the coefficient of $f$ modulo the $i$th prime.
Thus, monomials have to be stored only once and coefficient arithmetic can be efficiently vectorized using SIMD 
instructions.

\begin{example}\label{ex:rewrite}
Suppose we want to rewrite the specification
$
\spec = -5x + 4b + 3a
$
using the linear polynomial
$
f = 3x + b - a
$
over the three prime moduli
$
\mathcal{M} = (7, 11, 13).
$
We first represent $\spec$ componentwise as
$
\spec = (2,6,8)x
+ (4,4,4)b
+ (3,3,3)a,
$
since
$
-5 \equiv (2,6,8) \bmod{\mathcal{M}}.
$
Say $x$ is the leading monomial of $\spec$ and $f$.
Then we first normalize~$f$ by inverting its $x$-coefficient modulo each prime.
Multiplying $f$ componentwise by $(5,4,9)$ yields $f' = (1,1,1)x +(5,4,9) b + (2,7,4)a $.
Finally, we compute the following difference to cancel the $x$-term from $\spec$:
$$
	\spec - (2,6,8) \cdot f' \equiv (1,2,10)b + (6,5,10)a  \mod{\mathcal{M}}.
$$
\end{example}

If the result of such a linear rewriting step is zero, the circuit satisfies the specification and the procedure terminates.
Otherwise, we consider the leading monomial of the simplified specification and start the next iteration of the linear rewriting loop.
If we fail to retrieve a linear relation whose leading monomial matches that
 of the specification, we increase the depth of the subcircuit to be
extracted. To prevent the subcircuit from growing excessively, we switch to
nonlinear rewriting if no suitable linear relation is found after increasing the
subcircuit depth a fixed number of times. In \talismantwo, we
set this threshold to three, as determined empirically.

\medskip
\noindent
\textbf{Parallelization.}\;
A key property of our multimodular framework is that computations for different prime moduli are independent. 
Although this would allow the entire workflow in Figure~\ref{fig:flowchart} to be executed in parallel, 
such a coarse-grained strategy would duplicate substantial work, since several phases (e.g., preprocessing, subcircuit extraction, and sampling) 
are identical across primes.
Hence, we apply a more fine-grained parallelization scheme. 
Circuit encoding, preprocessing, and subcircuit extraction are done in serial and shared across primes, 
while the guess, prove, and repair steps are parallelized across moduli, with each thread handling one prime.
Sampling is parallelized as well, but independently of the primes. 
Here we use multiple threads to generate samples for the same subcircuit and aggregate them for guessing.  
Linear and nonlinear rewriting are executed serially, as further parallelization did not yield any speedups.

\subsection{Nonlinear Rewriting}
\label{ssec:nonlinear}

Once we switch from linear to nonlinear rewriting, we fix a lexicographic monomial
order that orders the variables following a reverse topological order of the circuit, starting from the output variables and
progressing toward the inputs. 
Under this ordering, each gate polynomial has a leading monomial consisting of a unique variable.
This implies that $G(C) \cup B(C)$ forms a Gröbner basis~\cite{KaufmannBiereKauers-FMSD19}.

\medskip
\noindent
\textbf{Simplify.}\; We first simplify the AIG to reduce the number of
rewriting steps. To this end, we eliminate intermediate nodes that have single fanouts.
Specifically, if a node $g$ only has one parent, 
we directly substitute the gate constraint of $g$ into that of its parent. 
We repeat this process until no such nodes remain. The gate polynomials of the rewritten AIG still form a Gröbner basis~\cite{KaufmannBiereKauers-FMSD19}.

\medskip
\noindent
\textbf{Rewrite.}\;
We iteratively rewrite the specification polynomial using the simplified gate polynomials, thereby computing the unique normal form modulo the
circuit ideal.
The fact that $G(C) \cup B(C)$ forms a Gröbner basis establishes the
soundness and completeness of nonlinear rewriting for deciding ideal membership. 
If the normal form of $\spec$ is the zero polynomial, then the
circuit fulfills the specification and is therefore correct.

Otherwise, $\spec \notin \langle G(C) \cup B(C) \rangle$. In this case, the
normal form consists only of primary inputs of the AIG. A counterexample can
therefore be extracted by choosing an input assignment for which the normal form
does not evaluate to zero. Deriving counterexamples is not immediately possible
in purely linear rewriting
approaches~\cite{Kaufmann-TACAS25,hofstadler_et_al:LIPIcs.CP.2025.14}, as the
normal form with respect to a degree-compatible order may still contain
internal circuit variables.

Nonlinear rewriting is performed analogously to the linear case
by representing coefficients modulo different primes by coefficient vectors (cf.~Example~\ref{ex:rewrite}).

\section{Experiments}
\label{sec:exp}

\begin{figure}[tb]
  \centering
    \includegraphics[width=0.85\columnwidth]{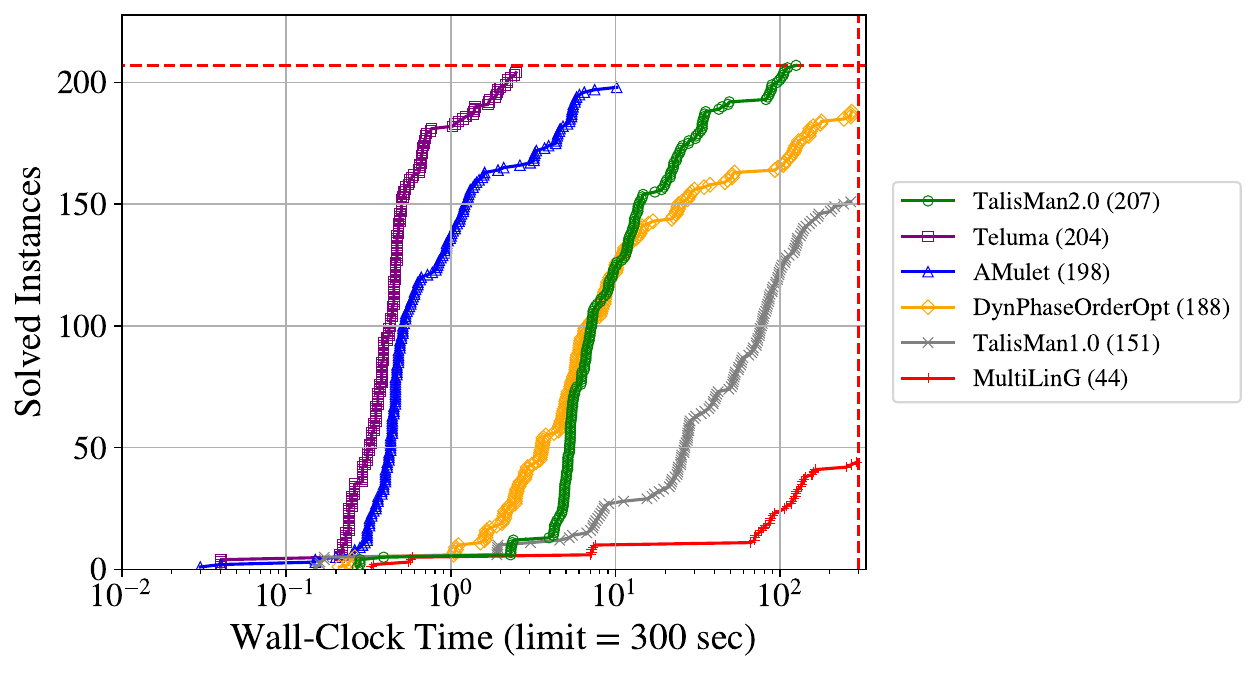}
      \caption{Comparison of \talismantwo and related work. For each tool the number of solved instances is added in parentheses.}
      \label{fig:related}
  \vspace{-3ex}
\end{figure}

We have implemented the presented hybrid multimodular approach in our tool called~\talismantwo~\cite{talismantwo}, 
which acts as an improved successor of~\talisman~\cite{hofstadler_et_al:LIPIcs.CP.2025.14}. 
In the following, we evaluate~\talismantwo~and compare it to related
work~\cite{hofstadler_et_al:LIPIcs.CP.2025.14,Kaufmann-TACAS25, KaufmannBiere-TACAS21,
KaufmannBeameBiereNordstrom-DATE22, KonradScholl-FMCAD24} using a total of 207 \emph{correct} integer multiplier benchmarks. 

We consider two types of benchmarks: \emph{structured
circuits}, where the components of a multiplier (partial product generation (PPG), partial product accumulation (PPA), and the final stage adder (FSA)) 
are clearly separable,
and \emph{synthesized circuits}, where gates are merged and rewritten to optimize the
circuit, blurring component boundaries and complicating direct verification. We
consider the following sets:

\begin{enumerate}
\item \emph{Structured aoki-multipliers}~\cite{aokipaper}: This set of 64-bit multipliers is generated by  
combining different architectures for PPG, PPA, and FSA (see Appendix~\ref{app:bench-details})
 yielding 192 structured multiplier architectures.

\item \emph{Synthesized \abc-multipliers}~\cite{abc}: 
We generate  32-, 64-, 128-bit multipliers of type \emph{sp-ar-rc} and optimize them with \abc
using four standard synthesis scripts (resyn, resyn2, resyn3, dc2)
and a complex optimization script\footnote{\tiny\texttt{-c "logic;
      mfs2 -W 20; ps; mfs; st; ps; dc2 -l; ps; resub -l -K 16 -N 3 -w
      100; ps; logic; mfs2 -W 20; ps; mfs; st; ps; iresyn -l; ps;
      resyn; ps; resyn2; ps; resyn3; ps; dc2 -l; ps;"}}. 
These 15 optimized benchmarks showcase the robustness of our approach. 
     
\end{enumerate}

In our experiments, we use primes of approximately 16 bits, 
which provide a good trade-off between keeping the number of required moduli small and enabling efficient coefficient encoding in the SAT solver.
Verifying 64-bit multipliers requires 8 such primes, while 128-bit multipliers require 16. 
All parallel computations use as many threads as we have primes.

The experiments were run  on dual-socket AMD EPYC 7313 @3.7GHz
machines running Ubuntu 24.04. Reported times are rounded wall-clock seconds. 
The time limit was \SI{300}{\second} and the memory limit was \SI{32}{\giga\byte} to ensure sufficient resources for all tools. 
Notably, \talismantwo itself only required at most~\SI{5}{\giga\byte}.

We compare \talismantwo against tools from related
work: \talisman~\cite{hofstadler_et_al:LIPIcs.CP.2025.14,talisman-artifact} and
\multiling~\cite{Kaufmann-TACAS25}, both of which also apply linear rewriting.
In addition, we consider \dynphaseorderopt~\cite{KonradScholl-FMCAD24},
\amulettwo~\cite{KaufmannBiereKauers-FMCAD19,KaufmannBiere-TACAS21}, and
\teluma~\cite{KaufmannBeameBiereNordstrom-DATE22}, all of which use nonlinear
rewriting based on a lexicographic monomial order. Figure~\ref{fig:related}
summarizes the comparison. 
Our main observations are:

\begin{itemize}
\item \talismantwo is the only tool that successfully solves all benchmarks. 
Increasing the time limit did not yield more successful runs for the other tools.

\item \amulettwo and \teluma solve the complete aoki benchmark set but fail on some of the synthesized ABC benchmarks. 
This supports our claim that their syntactic heuristics do not generalize well.

\item \dynphaseorderopt solves all 15 synthesized ABC benchmarks, but only 173 out of 192 aoki benchmarks. Notably, among the 19 benchmarks it fails on, 14 contain a carry look-ahead adder (cl) as FSA, indicating that this structure 
is hard for purely nonlinear rewriting approaches.  

\end{itemize}

\begin{figure}[tb]
  \centering 
  \begin{minipage}{0.32\textwidth}
      \centering
       \includegraphics[width=0.95\columnwidth]{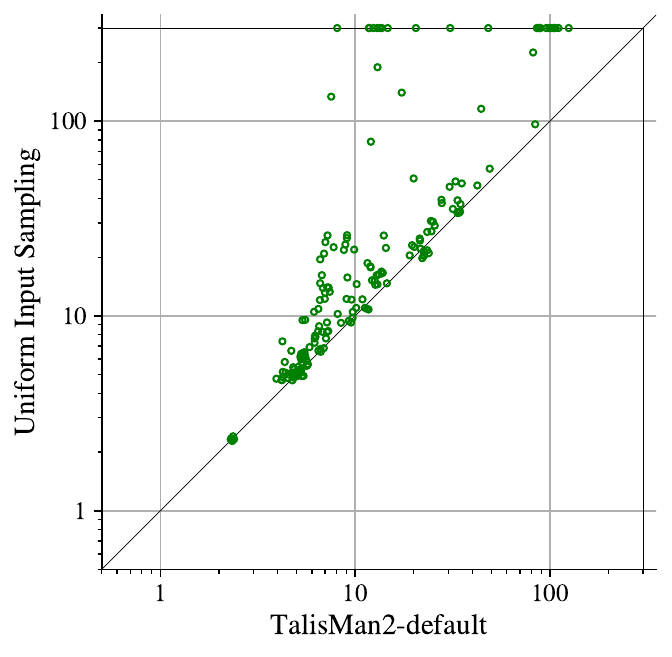}

  \end{minipage}
  \hfill
  \begin{minipage}{0.32\textwidth}
    \centering
       \includegraphics[width=0.95\columnwidth]{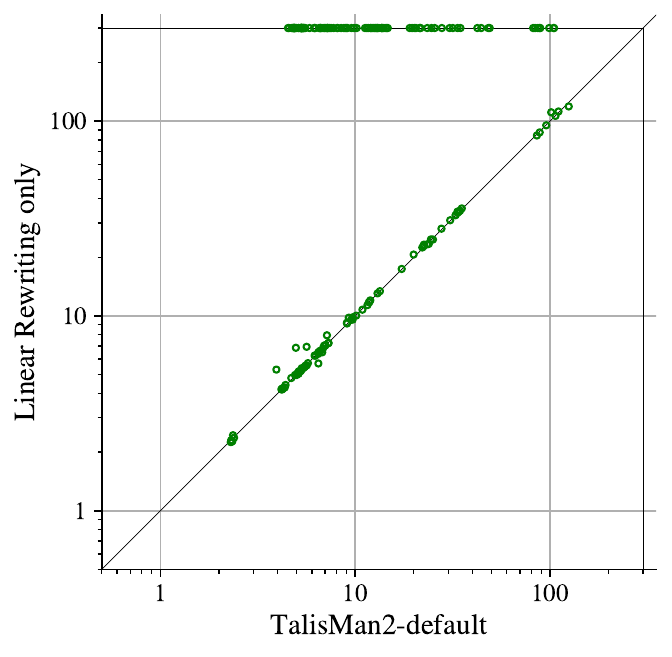}
       
  \end{minipage}
    \hfill
  \begin{minipage}{0.32\textwidth}
    \centering
       \includegraphics[width=0.95\columnwidth]{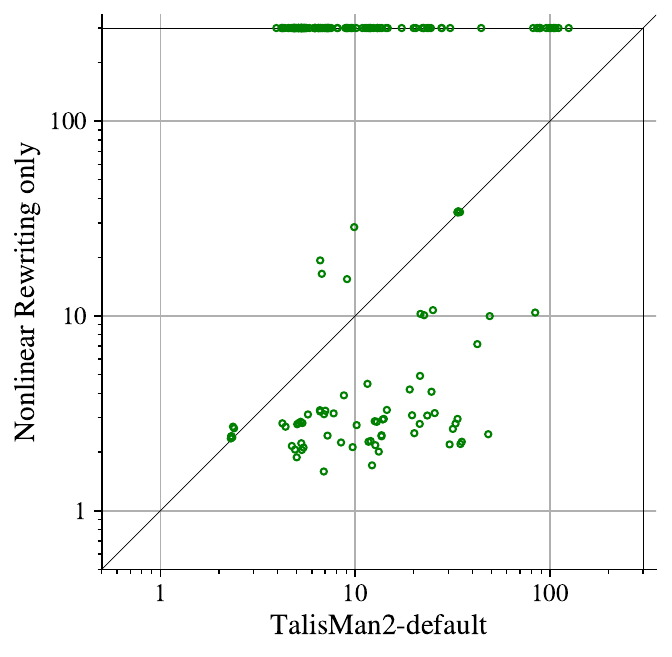}
  \end{minipage}
  \caption{Comparing the default setting when either uniform input sampling (left), only linear rewriting (middle), or only nonlinear rewriting is used (right).}
  \label{fig:ablation}
  \vspace{-1ex}
\end{figure}

\renewcommand{\tabcolsep}{2.5pt}
{\renewcommand{\arraystretch}{1.1}
\begin{table}[tb]
  \centering
  \scriptsize
  \begin{tabular}{l|r |r|r|r|r|r|r|r||r|r }
    \multirow{3}{*}{Name} & \multirow{3}{*}{Total}
      & \multicolumn{9}{c}{\talismantwo\ --- Phase Times (s)} \\
       \cline{3-11}
      && \multicolumn{7}{c||}{Linear} & \multicolumn{2}{c}{Nonlinear}  \\
    \cline{3-11}
      && Prepr.
      & Extract 
      & Sample 
      & Guess* 
      & Prove*
      & Repair* 
      & Rewrite
      & Simplify 
      & Rewrite  \\
    \midrule
abc64-cmp & 2.3 & 2.1 & <0.05& <0.05& <0.05& <0.05& - & <0.05& - & - \\   
    \midrule
bp-ar-cl & 44.0 & 1.4 & 0.7 & 0.9 & 9.7 & 26.6 & 3.0 & 0.1 & <0.05& 0.2 \\
bp-ba-cn & 42.0 & 1.6 & <0.05& 2.3 & 23.0 & 3.6 & 5.3 & <0.05& 4.8 & 0.4 \\
bp-bd-cl & 86.2 & 1.4 & 0.9 & 1.6 & 21.6 & 50.1 & 8.3 & 0.2 & <0.05& 0.2 \\
bp-cn-cs & 13.5 & 1.5 & <0.05& 1.4 & 7.9 & 1.6 & - & <0.05& 0.2 & 0.4 \\
bp-ct-hc & 9.0 & 0.9 & 2.1 & 0.2 & 3.9 & 1.1 & - & 0.1 & <0.05& 0.2 \\
bp-dt-lf & 7.1 & 1.3 & 0.5 & 0.2 & 3.3 & 1.3 & - & 0.1 & <0.05& 0.2 \\
bp-wt-rc & 4.8 & 1.4 & 0.5 & 0.1 & 1.1 & 1.2 & - & 0.1 & <0.05& 0.2 \\
\midrule 
sp-ar-cl & 19.7 & 2.2 & 0.1 & 0.4 & 3.1 & 12.6 & 0.6 & 0.1 & - & - \\
sp-ba-cn & 48.6 & 2.7 & 0.1 & 2.8 & 25.0 & 7.1 & 2.8 & <0.05& 6.7 & 0.3 \\
sp-bd-cl & 88.1 & 2.3 & 0.3 & 1.6 & 19.8 & 54.4 & 7.0 & 0.2 & - & - \\
sp-cn-cs & 21.4 & 2.7 & <0.05& 1.3 & 7.9 & 1.8 & - & <0.05& 0.3 & 6.8 \\
sp-ct-hc & 11.9 & 1.6 & 4.2 & 0.2 & 4.1 & 1.3 & - & 0.2 & - & - \\
sp-dt-lf & 6.8 & 2.1 & <0.05& 0.2 & 2.8 & 1.4 & - & 0.1 & - & - \\
sp-wt-rc & 5.0 & 2.3 & <0.05 & 0.2 & 0.9 & 1.4 & - & 0.1 & - & - \\
    \bottomrule
  \end{tabular}
  \caption{Per-phase runtime (rounded to tenths of seconds) of \talismantwo. 
  In columns marked with `*' we report the average time per prime modulus.
  Entries `-' indicate that this phase was not entered. 
  }
  \label{tbl:statistics}
  \vspace{-5ex}
\end{table}
}

We highlight the runtimes of the different phases of \talismantwo from  Figure~\ref{fig:flowchart} in Table~\ref{tbl:statistics}.
A complete table is contained in Appendix~\ref{app:exp-details}.
Notably, we only had to enter the repair phase in 36 out of 207 benchmarks.
In all other cases the guesses were immediately correct. 
Almost all of those benchmarks have either a carry look-ahead (cl) or a conditional sum adder (cn) as the FSA.

To assess the contribution of individual components of \talismantwo, we conducted
an ablation study in which selected improvements were disabled. The results
are shown in Figure~\ref{fig:ablation}. When disabling \cmsgen-based sampling and
instead using uniform input sampling as in \talisman (left plot), \talismantwo
fails to solve 24 benchmarks and exhibits increased runtimes on almost all
remaining instances. Restricting the approach to linear rewriting only (middle
plot) reduces the number of solved benchmarks to 98. 
Conversely, when only nonlinear
rewriting is applied (right plot), 79 benchmarks can be solved; however, within
this subset, \talismantwo is faster than the hybrid configuration in 63 cases,
indicating that nonlinear rewriting is particularly effective on selected
instances but lacks overall robustness when used in isolation.

We evaluated the impact of replacing our custom linear algebra routines by \flint, which increases the overall runtime by 34\% on average.
We also investigated whether using (pseudo-)Mersenne primes, i.e., primes of the form $2^{k} - c$ with small $c$, improves performance in the proving phase. 
Our intuition was that guessed linear relations often contain coefficients that are powers of two. 
Modulo (pseudo-)Mersenne primes, such coefficients map to small residues, which could simplify the SAT-based verification
as smaller integers lead to smaller propositional encodings. 
To test this conjecture, we compared the default configuration of \talismantwo using only pseudo-Mersenne primes in the range $2^{16} \pm 50$ with a configuration using randomly selected primes from the interval $[2^{15}, 2^{17}]$. The latter exhibited only a negligible average slowdown of $1.2\%$.
We have also evaluated the impact of different prime sizes. Using 8-bit primes results in an average slowdown of $27\%$ compared to the 16-bit default, likely due to increased memory traffic from more parallel threads. In contrast, 31-bit primes lead to a slowdown of $394\%$, with only 161 out of 207 benchmarks solved within the \SI{300}{\second} time limit. 
This significant degradation is caused by substantially higher SAT verification costs due to larger coefficients in the guesses.

\section{Conclusion \& Future Work}
We presented a hybrid multimodular framework for word-level verification of
arithmetic circuits that completely avoids arbitrary-precision arithmetic. By
performing algebraic reasoning in parallel over word-sized prime fields and
combining linear with nonlinear rewriting, our approach achieves both efficiency and robustness. 
Our implementation in \talismantwo shows that this design significantly improves performance in practice. 

An important next step is the integration of proof logging techniques such as~\cite{KFBK24,KH25} into our framework. 
In particular, recycling proof steps as proposed in~\cite{KH25} appears especially promising in the multimodular setting. 
We also plan to investigate methods for generating short proofs of ideal membership~\cite{HV24}. 
Beyond circuit verification, we expect that our techniques can be applied also to related problems such as equivalence checking.


\begin{credits}
  \subsubsection{\ackname}
  C.H. was supported by the LIT AI Lab funded by the state of Upper
  Austria and by the Austrian Science Fund (FWF) [10.55776/COE12].
   D.K. was supported by the Austrian Science Fund (FWF) under
  grant [10.55776/ESP666]. 

  \subsubsection{\discintname}
  The authors have no competing interests to declare that are relevant to the
  content of this article.
\end{credits}

%
%
%
\bibliographystyle{splncs04}
\bibliography{references}
\newpage
\appendix

\section{Benchmark Details}\label{app:bench-details}

The components of the \emph{aoki}-benchmarks are: 
\begin{itemize}
\item PPG: simple (sp), Booth encoding radix-4 (bp);
\item PPA: Array                         (ar),
Wallace tree                  (wt),
Balanced delay tree           (bd),
Overturned-stairs tree        (os),
Dadda tree                    (dt),
(4;2) compressor tree         (ct),
(7,3) counter tree            (cn),
Red. binary addition tree     (ba);
\item FSA:
Ripple-carry                  (rc),
  Carry look-ahead              (cl),
  Ripple-block carry look-ahead (rb),
 Block carry look-ahead        (bc),
  Ladner-Fischer                (lf),
  Kogge-Stone                   (ks),
  Brent-Kung                    (bk),
  Han-Carlson                   (hc),
 Conditional sum               (cn),
  Carry select                  (cs),
  Carry-skip fix size           (csf),
  Carry-skip var. size          (csv)
\end{itemize}
 All of these circuits have an input bit-width of  64 and consist of \num{38 000} to \num{52 000} nodes.

\section{Experimental Data}\label{app:exp-details}

\renewcommand{\tabcolsep}{3pt}
{\renewcommand{\arraystretch}{1.2}
 \scriptsize
\begin{longtable}{l|r|r|r|r|r|r|r|r||r|r}

  \caption{Per-phase runtime (rounded to tenths of seconds) of \talismantwo. 
  In columns marked with `*' we report the average time per prime modulus.
  Entries `-' indicate that this phase was not entered. 
  }\\
  
    \multirow{3}{*}{Name} & \multirow{3}{*}{Total }
      & \multicolumn{9}{c}{\talismantwo\ --- Phase Times (s)} \\
       \cline{3-11}
      && \multicolumn{7}{c||}{Linear} & \multicolumn{2}{c}{Nonlinear}  \\
    \cline{3-11}
      && Prepr.
      & Extract 
      & Sample 
      & Guess* 
      & Prove*
      & Repair* 
      & Rewrite
      & Simplify 
      & Rewrite  \\
    \midrule
    \endhead

    \bottomrule
    \endfoot 
    bp-ar-bc & 5.2 & 1.5 & 0.7 & 0.1 & 1.1 & 1.3 & - & 0.1 & <0.05& 0.2 \\
bp-ar-bk & 4.7 & 1.5 & 0.5 & 0.1 & 1.0 & 1.1 & - & 0.1 & <0.05& 0.2 \\
bp-ar-cl & 44.0 & 1.4 & 0.7 & 0.9 & 9.7 & 26.6 & 3.0 & 0.1 & <0.05& 0.2 \\
bp-ar-cn & 48.0 & 1.6 & <0.05& 2.3 & 26.7 & 3.2 & 12.8 & <0.05& 0.4 & 0.3 \\
bp-ar-cs & 8.4 & 1.5 & <0.05& 0.7 & 4.2 & 1.2 & - & <0.05& 0.2 & 0.1 \\
bp-ar-csf & 4.4 & 1.5 & 0.4 & 0.1 & 0.7 & 1.2 & - & 0.1 & <0.05& 0.2 \\
bp-ar-csv & 4.5 & 1.5 & 0.5 & 0.1 & 0.7 & 1.2 & - & 0.1 & <0.05& 0.2 \\
bp-ar-hc & 5.5 & 1.5 & 0.4 & 0.2 & 1.9 & 0.9 & - & 0.1 & <0.05& 0.2 \\
bp-ar-ks & 8.0 & 1.5 & 0.5 & 0.3 & 4.1 & 1.1 & - & 0.1 & <0.05& 0.2 \\
bp-ar-lf & 5.3 & 1.4 & 0.5 & 0.1 & 1.6 & 1.0 & - & 0.1 & <0.05& 0.2 \\
bp-ar-rb & 5.2 & 1.5 & 0.5 & 0.2 & 1.4 & 1.2 & - & 0.1 & <0.05& 0.2 \\
bp-ar-rc & 4.7 & 1.5 & 0.4 & 0.1 & 0.9 & 1.1 & - & 0.1 & <0.05& 0.2 \\
bp-ba-bc & 5.7 & 1.5 & 0.5 & 0.2 & 1.7 & 1.4 & - & 0.1 & <0.05& 0.2 \\
bp-ba-bk & 6.1 & 1.4 & 0.5 & 0.2 & 2.1 & 1.2 & - & 0.1 & <0.05& 0.2 \\
bp-ba-cl & 104.0 & 1.5 & 1.0 & 1.9 & 26.3 & 61.3 & 9.6 & 0.2 & <0.05& 0.2 \\
bp-ba-cn & 42.0 & 1.6 & <0.05& 2.3 & 23.0 & 3.6 & 5.3 & <0.05& 4.8 & 0.4 \\
bp-ba-cs & 13.5 & 1.6 & <0.05& 1.3 & 8.0 & 1.6 & - & <0.05& 0.4 & 0.1 \\
bp-ba-csf & 5.1 & 1.5 & 0.5 & 0.1 & 1.0 & 1.3 & - & 0.1 & <0.05& 0.2 \\
bp-ba-csv & 5.0 & 1.5 & 0.5 & 0.1 & 1.0 & 1.3 & - & 0.1 & <0.05& 0.2 \\
bp-ba-hc & 8.8 & 1.5 & 0.6 & 0.3 & 4.4 & 1.4 & - & 0.1 & <0.05& 0.2 \\
bp-ba-ks & 27.6 & 1.5 & 1.0 & 2.0 & 15.8 & 5.6 & 0.2 & 0.1 & <0.05& 0.2 \\
bp-ba-lf & 9.8 & 1.5 & 0.5 & 0.3 & 4.3 & 1.4 & 1.1 & 0.1 & <0.05& 0.2 \\
bp-ba-rb & 7.1 & 1.5 & 0.5 & 0.2 & 2.5 & 1.3 & - & 0.1 & <0.05& 0.2 \\
bp-ba-rc & 5.2 & 1.5 & 0.5 & 0.2 & 1.2 & 1.3 & - & 0.1 & <0.05& 0.2 \\
bp-bd-bc & 5.2 & 1.4 & 0.5 & 0.2 & 1.4 & 1.2 & - & 0.1 & <0.05& 0.2 \\
bp-bd-bk & 5.2 & 1.4 & 0.5 & 0.1 & 1.5 & 1.1 & - & 0.1 & <0.05& 0.2 \\
bp-bd-cl & 86.2 & 1.4 & 0.9 & 1.6 & 21.6 & 50.1 & 8.3 & 0.2 & <0.05& 0.2 \\
bp-bd-cn & 24.4 & 1.5 & 0.1 & 1.6 & 11.3 & 5.9 & 0.9 & <0.05& 1.9 & 0.4 \\
bp-bd-cs & 30.3 & 1.5 & 0.5 & 3.2 & 22.2 & 1.9 & - & 0.1 & <0.05& 0.2 \\
bp-bd-csf & 4.8 & 1.4 & 0.5 & 0.1 & 1.0 & 1.3 & - & 0.1 & <0.05& 0.2 \\
bp-bd-csv & 4.8 & 1.4 & 0.5 & 0.1 & 1.0 & 1.3 & - & 0.1 & <0.05& 0.2 \\
bp-bd-hc & 6.9 & 1.4 & 0.5 & 0.2 & 3.1 & 1.1 & - & 0.1 & <0.05& 0.2 \\
bp-bd-ks & 12.3 & 1.4 & 0.5 & 0.6 & 7.3 & 1.9 & - & 0.1 & <0.05& 0.2 \\
bp-bd-lf & 6.6 & 1.4 & 0.5 & 0.2 & 2.7 & 1.3 & - & 0.1 & <0.05& 0.2 \\
bp-bd-rb & 6.1 & 1.4 & 0.5 & 0.2 & 2.2 & 1.2 & - & 0.1 & <0.05& 0.2 \\
bp-bd-rc & 5.2 & 1.4 & 0.7 & 0.2 & 1.2 & 1.2 & - & 0.1 & <0.05& 0.2 \\
bp-cn-bc & 9.4 & 1.4 & 1.2 & 0.2 & 2.8 & 2.8 & - & 0.1 & <0.05& 0.2 \\
bp-cn-bk & 8.9 & 1.4 & 1.2 & 0.2 & 2.8 & 2.5 & - & 0.1 & <0.05& 0.2 \\
bp-cn-cl & 104.5 & 1.4 & 3.0 & 1.9 & 27.7 & 58.5 & 9.0 & 0.3 & <0.05& 0.2 \\
bp-cn-cn & 31.3 & 1.5 & <0.05& 1.9 & 19.1 & 4.0 & 2.7 & <0.05& 0.4 & 0.6 \\
bp-cn-cs & 13.5 & 1.5 & <0.05& 1.4 & 7.9 & 1.6 & - & <0.05& 0.2 & 0.4 \\
bp-cn-csf & 11.1 & 1.4 & 2.3 & 0.1 & 3.4 & 2.9 & - & 0.2 & <0.05& 0.2 \\
bp-cn-csv & 11.3 & 1.4 & 2.3 & 0.2 & 3.5 & 3.0 & - & 0.2 & <0.05& 0.2 \\
bp-cn-hc & 14.2 & 1.4 & 2.2 & 0.3 & 6.3 & 2.9 & - & 0.2 & <0.05& 0.2 \\
bp-cn-ks & 20.3 & 1.4 & 2.3 & 0.8 & 11.0 & 3.7 & - & 0.2 & <0.05& 0.2 \\
bp-cn-lf & 13.9 & 1.4 & 2.3 & 0.3 & 5.8 & 3.0 & - & 0.2 & <0.05& 0.2 \\
bp-cn-rb & 12.9 & 1.4 & 2.2 & 0.3 & 5.0 & 3.0 & - & 0.2 & <0.05& 0.2 \\
bp-cn-rc & 11.6 & 1.4 & 2.2 & 0.2 & 3.8 & 2.9 & - & 0.2 & <0.05& 0.2 \\
bp-ct-bc & 7.2 & 0.9 & 2.1 & 0.2 & 2.1 & 1.2 & - & 0.1 & <0.05& 0.2 \\
bp-ct-bk & 7.0 & 0.9 & 2.1 & 0.2 & 2.1 & 1.1 & - & 0.1 & <0.05& 0.2 \\
bp-ct-cl & 88.6 & 0.9 & 2.8 & 1.7 & 22.7 & 48.5 & 8.8 & 0.2 & <0.05& 0.2 \\
bp-ct-cn & 18.8 & 1.0 & <0.05& 1.1 & 8.9 & 3.1 & 0.9 & <0.05& 2.6 & 0.4 \\
bp-ct-cs & 12.1 & 1.0 & <0.05& 1.3 & 7.4 & 1.5 & - & <0.05& 0.2 & 0.2 \\
bp-ct-csf & 6.6 & 0.9 & 2.1 & 0.1 & 1.6 & 1.2 & - & 0.1 & <0.05& 0.2 \\
bp-ct-csv & 6.6 & 1.0 & 2.1 & 0.1 & 1.6 & 1.2 & - & 0.1 & <0.05& 0.2 \\
bp-ct-hc & 9.0 & 0.9 & 2.1 & 0.2 & 3.9 & 1.1 & - & 0.1 & <0.05& 0.2 \\
bp-ct-ks & 14.6 & 0.9 & 2.1 & 0.7 & 8.3 & 1.8 & - & 0.1 & <0.05& 0.2 \\
bp-ct-lf & 8.7 & 0.9 & 2.1 & 0.2 & 3.5 & 1.3 & - & 0.1 & <0.05& 0.2 \\
bp-ct-rb & 8.1 & 0.9 & 2.2 & 0.2 & 2.9 & 1.2 & - & 0.1 & <0.05& 0.2 \\
bp-ct-rc & 6.8 & 1.0 & 2.2 & 0.2 & 1.8 & 1.1 & - & 0.1 & <0.05& 0.2 \\
bp-dt-bc & 5.4 & 1.3 & 0.5 & 0.2 & 1.6 & 1.3 & - & 0.1 & <0.05& 0.2 \\
bp-dt-bk & 5.2 & 1.3 & 0.5 & 0.2 & 1.6 & 1.2 & - & 0.1 & <0.05& 0.2 \\
bp-dt-cl & 98.6 & 1.4 & 1.4 & 1.9 & 25.4 & 56.8 & 8.9 & 0.2 & <0.05& 0.2 \\
bp-dt-cn & 34.4 & 1.4 & <0.05& 2.2 & 22.0 & 3.7 & 3.3 & <0.05& 0.4 & 0.4 \\
bp-dt-cs & 13.1 & 1.4 & <0.05& 1.3 & 8.0 & 1.7 & - & <0.05& 0.2 & 0.1 \\
bp-dt-csf & 4.7 & 1.3 & 0.5 & 0.1 & 1.0 & 1.3 & - & 0.1 & <0.05& 0.2 \\
bp-dt-csv & 4.8 & 1.3 & 0.5 & 0.1 & 1.0 & 1.3 & - & 0.1 & <0.05& 0.2 \\
bp-dt-hc & 7.5 & 1.3 & 0.5 & 0.3 & 3.6 & 1.2 & - & 0.1 & <0.05& 0.2 \\
bp-dt-ks & 13.6 & 1.2 & 0.5 & 0.7 & 8.5 & 2.0 & - & 0.1 & <0.05& 0.2 \\
bp-dt-lf & 7.1 & 1.3 & 0.5 & 0.2 & 3.3 & 1.3 & - & 0.1 & <0.05& 0.2 \\
bp-dt-rb & 5.4 & 1.3 & 0.5 & 0.2 & 1.6 & 1.3 & - & 0.1 & <0.05& 0.2 \\
bp-dt-rc & 4.9 & 1.3 & 0.5 & 0.2 & 1.2 & 1.2 & - & 0.1 & <0.05& 0.2 \\
bp-os-bc & 5.3 & 1.4 & 0.5 & 0.2 & 1.5 & 1.2 & - & 0.1 & <0.05& 0.2 \\
bp-os-bk & 5.3 & 1.4 & 0.4 & 0.2 & 1.6 & 1.2 & - & 0.1 & <0.05& 0.2 \\
bp-os-cl & 88.2 & 1.4 & 0.9 & 1.6 & 21.9 & 52.1 & 8.0 & 0.2 & <0.05& 0.2 \\
bp-os-cn & 21.2 & 1.5 & 0.1 & 1.5 & 9.8 & 5.8 & 0.7 & <0.05& 0.8 & 0.4 \\
bp-os-cs & 12.6 & 1.5 & <0.05& 1.2 & 7.3 & 1.7 & - & <0.05& 0.3 & 0.1 \\
bp-os-csf & 4.8 & 1.4 & 0.5 & 0.1 & 1.0 & 1.3 & - & 0.1 & <0.05& 0.2 \\
bp-os-csv & 4.9 & 1.4 & 0.5 & 0.1 & 1.0 & 1.3 & - & 0.1 & <0.05& 0.2 \\
bp-os-hc & 7.3 & 1.4 & 0.5 & 0.3 & 3.5 & 1.1 & - & 0.1 & <0.05& 0.2 \\
bp-os-ks & 12.9 & 1.4 & 0.5 & 0.6 & 7.7 & 2.0 & - & 0.1 & <0.05& 0.2 \\
bp-os-lf & 7.7 & 1.4 & 0.5 & 0.2 & 2.9 & 1.3 & 0.7 & 0.1 & <0.05& 0.2 \\
bp-os-rb & 6.5 & 1.4 & 0.5 & 0.2 & 2.4 & 1.3 & - & 0.1 & <0.05& 0.2 \\
bp-os-rc & 5.4 & 1.4 & 0.8 & 0.2 & 1.3 & 1.2 & - & 0.1 & <0.05& 0.2 \\
bp-wt-bc & 5.2 & 1.4 & 0.5 & 0.1 & 1.4 & 1.3 & - & 0.1 & <0.05& 0.2 \\
bp-wt-bk & 5.2 & 1.4 & 0.5 & 0.1 & 1.4 & 1.1 & - & 0.1 & <0.05& 0.2 \\
bp-wt-cl & 81.5 & 1.4 & 0.9 & 1.6 & 20.4 & 47.9 & 6.9 & 0.1 & <0.05& 0.2 \\
bp-wt-cn & 21.3 & 1.5 & 0.1 & 1.4 & 8.7 & 5.3 & 0.6 & <0.05& 2.7 & 0.4 \\
bp-wt-cs & 11.9 & 1.5 & <0.05& 1.1 & 6.8 & 1.5 & - & <0.05& 0.3 & 0.1 \\
bp-wt-csf & 4.8 & 1.4 & 0.5 & 0.1 & 0.9 & 1.3 & - & 0.1 & <0.05& 0.2 \\
bp-wt-csv & 4.8 & 1.4 & 0.5 & 0.1 & 0.9 & 1.3 & - & 0.1 & <0.05& 0.2 \\
bp-wt-hc & 7.1 & 1.4 & 0.5 & 0.2 & 3.1 & 1.3 & - & 0.1 & <0.05& 0.2 \\
bp-wt-ks & 11.9 & 1.4 & 0.5 & 0.6 & 7.2 & 1.6 & - & 0.1 & <0.05& 0.2 \\
bp-wt-lf & 6.6 & 1.4 & 0.5 & 0.2 & 2.6 & 1.3 & - & 0.1 & <0.05& 0.2 \\
bp-wt-rb & 6.1 & 1.4 & 0.5 & 0.2 & 2.2 & 1.2 & - & 0.1 & <0.05& 0.2 \\
bp-wt-rc & 4.8 & 1.4 & 0.5 & 0.1 & 1.1 & 1.2 & - & 0.1 & <0.05& 0.2 \\
sp-ar-bc & 4.2 & 2.2 & <0.05& 0.1 & 0.4 & 1.2 & - & <0.05& - & - \\
sp-ar-bk & 3.9 & 2.1 & <0.05& 0.1 & 0.4 & 1.1 & - & <0.05& - & - \\
sp-ar-cl & 19.7 & 2.2 & 0.1 & 0.4 & 3.1 & 12.6 & 0.6 & 0.1 & - & - \\
sp-ar-cn & 10.0& 2.2 & <0.05& 0.5 & 4.2 & 1.6 & 0.7 & <0.05& 0.3 & 0.1 \\
sp-ar-cs & 4.3 & 2.2 & <0.05& 0.2 & 1.1 & 0.7 & - & <0.05& - & - \\
sp-ar-csf & 4.1 & 2.2 & <0.05& <0.05& 0.3 & 1.4 & - & <0.05& - & - \\
sp-ar-csv & 4.2 & 2.2 & <0.05& 0.1 & 0.3 & 1.4 & - & <0.05& - & - \\
sp-ar-hc & 4.3 & 2.1 & <0.05& 0.1 & 0.7 & 1.0 & - & <0.05& - & - \\
sp-ar-ks & 4.6 & 2.1 & <0.05& 0.1 & 1.3 & 0.8 & - & <0.05& - & - \\
sp-ar-lf & 4.2 & 2.1 & <0.05& 0.1 & 0.7 & 1.0 & - & <0.05& - & - \\
sp-ar-rb & 4.2 & 2.1 & <0.05& 0.1 & 0.6 & 1.1 & - & <0.05& - & - \\
sp-ar-rc & 2.4 & 2.2 & <0.05& <0.05& <0.05& <0.05& - & <0.05& - & - \\
sp-ba-bc & 6.1 & 2.5 & 0.1 & 0.1 & 1.4 & 1.6 & - & 0.1 & - & - \\
sp-ba-bk & 6.4 & 2.5 & 0.1 & 0.2 & 1.9 & 1.4 & - & 0.1 & - & - \\
sp-ba-cl & 110.1 & 2.6 & 0.5 & 1.8 & 25.0 & 68.1 & 10.1 & 0.1 & - & - \\
sp-ba-cn & 48.6 & 2.7 & 0.1 & 2.8 & 25.0 & 7.1 & 2.8 & <0.05& 6.7 & 0.3 \\
sp-ba-cs & 14.4 & 2.7 & <0.05& 1.3 & 7.8 & 1.8 & - & <0.05& 0.3 & 0.1 \\
sp-ba-csf & 5.6 & 2.5 & 0.2 & 0.1 & 0.9 & 1.6 & - & 0.1 & - & - \\
sp-ba-csv & 5.6 & 2.5 & 0.2 & 0.1 & 0.9 & 1.6 & - & 0.1 & - & - \\
sp-ba-hc & 9.0 & 2.6 & 0.1 & 0.3 & 4.3 & 1.4 & - & 0.1 & - & - \\
sp-ba-ks & 27.4 & 2.5 & 0.3 & 1.8 & 15.1 & 6.4 & 0.2 & 0.1 & - & - \\
sp-ba-lf & 9.0 & 2.5 & 0.1 & 0.3 & 4.1 & 1.6 & - & 0.1 & - & - \\
sp-ba-rb & 7.1 & 2.5 & 0.1 & 0.2 & 2.4 & 1.5 & - & 0.1 & - & - \\
sp-ba-rc & 5.6 & 2.6 & 0.1 & 0.2 & 1.0 & 1.5 & - & 0.1 & - & - \\
sp-bd-bc & 5.3 & 2.3 & <0.05& 0.1 & 1.2 & 1.4 & - & 0.1 & - & - \\
sp-bd-bk & 5.2 & 2.3 & <0.05& 0.1 & 1.2 & 1.3 & - & 0.1 & - & - \\
sp-bd-cl & 88.1 & 2.3 & 0.3 & 1.6 & 19.8 & 54.4 & 7.0 & 0.2 & - & - \\
sp-bd-cn & 23.2 & 2.4 & 0.1 & 1.5 & 10.1 & 7.0 & 0.7 & <0.05& 0.4 & 0.3 \\
sp-bd-cs & 12.8 & 2.4 & <0.05& 1.1 & 6.7 & 1.7 & - & <0.05& 0.3 & 0.1 \\
sp-bd-csf & 4.9 & 2.3 & <0.05& 0.1 & 0.7 & 1.5 & - & <0.05& - & - \\
sp-bd-csv & 5.0 & 2.3 & <0.05& 0.1 & 0.8 & 1.6 & - & 0.1 & - & - \\
sp-bd-hc & 6.8 & 2.3 & <0.05& 0.2 & 2.9 & 1.0 & - & 0.1 & - & - \\
sp-bd-ks & 11.7 & 2.3 & <0.05& 0.6 & 6.7 & 1.7 & - & 0.1 & - & - \\
sp-bd-lf & 6.5 & 2.3 & <0.05& 0.2 & 2.4 & 1.4 & - & 0.1 & - & - \\
sp-bd-rb & 6.2 & 2.3 & <0.05& 0.2 & 1.9 & 1.4 & - & 0.1 & - & - \\
sp-bd-rc & 5.0 & 2.3 & <0.05& 0.1 & 0.9 & 1.4 & - & <0.05& - & - \\
sp-cn-bc & 22.2 & 2.7 & 4.2 & 0.3 & 6.7 & 6.8 & - & 0.3 & - & - \\
sp-cn-bk & 23.2 & 2.7 & 4.3 & 0.3 & 7.4 & 6.8 & - & 0.3 & - & - \\
sp-cn-cl & 124.3 & 2.6 & 5.3 & 1.9 & 30.5 & 71.7 & 9.0 & 0.4 & - & - \\
sp-cn-cn & 83.3 & 2.7 & 0.1 & 4.5 & 53.1 & 7.1 & 6.9 & <0.05& 0.4 & 6.9 \\
sp-cn-cs & 21.4 & 2.7 & <0.05& 1.3 & 7.9 & 1.8 & - & <0.05& 0.3 & 6.8 \\
sp-cn-csf & 21.9 & 2.6 & 4.4 & 0.2 & 6.1 & 6.9 & - & 0.3 & - & - \\
sp-cn-csv & 22.3 & 2.6 & 4.2 & 0.2 & 6.4 & 7.3 & - & 0.3 & - & - \\
sp-cn-hc & 24.2 & 2.7 & 4.3 & 0.4 & 8.6 & 6.6 & - & 0.3 & - & - \\
sp-cn-ks & 30.5 & 2.6 & 4.2 & 0.8 & 13.5 & 7.7 & - & 0.3 & - & - \\
sp-cn-lf & 24.9 & 2.6 & 4.2 & 0.3 & 8.4 & 7.2 & - & 0.3 & - & - \\
sp-cn-rb & 23.6 & 2.6 & 4.3 & 0.3 & 7.8 & 6.9 & - & 0.3 & - & - \\
sp-cn-rc & 22.4 & 2.7 & 4.2 & 0.3 & 6.4 & 7.0 & - & 0.3 & - & - \\
sp-ct-bc & 10.0& 1.7 & 4.2 & 0.2 & 2.2 & 1.4 & - & 0.2 & - & - \\
sp-ct-bk & 9.6 & 1.6 & 4.0 & 0.2 & 2.2 & 1.2 & - & 0.2 & - & - \\
sp-ct-cl & 95.2 & 1.7 & 5.2 & 1.7 & 22.7 & 53.6 & 8.1 & 0.3 & - & - \\
sp-ct-cn & 19.8 & 1.8 & 0.1 & 1.4 & 8.6 & 6.1 & 0.5 & <0.05& 0.4 & 0.4 \\
sp-ct-cs & 35.1 & 1.8 & 4.1 & 3.2 & 23.0 & 1.9 & - & 0.2 & - & - \\
sp-ct-csf & 9.4 & 1.7 & 4.2 & 0.1 & 1.7 & 1.4 & - & 0.2 & - & - \\
sp-ct-csv & 9.2 & 1.6 & 4.2 & 0.1 & 1.5 & 1.4 & - & 0.2 & - & - \\
sp-ct-hc & 11.9 & 1.6 & 4.2 & 0.2 & 4.1 & 1.3 & - & 0.2 & - & - \\
sp-ct-ks & 17.2 & 1.7 & 4.1 & 0.6 & 8.3 & 2.0 & - & 0.2 & - & - \\
sp-ct-lf & 11.5 & 1.7 & 4.2 & 0.2 & 3.5 & 1.4 & - & 0.2 & - & - \\
sp-ct-rb & 10.8 & 1.6 & 4.1 & 0.2 & 3.1 & 1.4 & - & 0.2 & - & - \\
sp-ct-rc & 9.6 & 1.7 & 4.2 & 0.1 & 1.9 & 1.3 & - & 0.2 & - & - \\
sp-dt-bc & 5.4 & 2.2 & <0.05& 0.2 & 1.4 & 1.4 & - & 0.1 & - & - \\
sp-dt-bk & 5.2 & 2.2 & <0.05& 0.2 & 1.3 & 1.3 & - & 0.1 & - & - \\
sp-dt-cl & 106.3 & 2.1 & 0.4 & 1.7 & 24.8 & 65.5 & 9.5 & 0.1 & - & - \\
sp-dt-cn & 33.3 & 2.2 & 0.1 & 2.5 & 18.3 & 7.0 & 1.5 & <0.05& 0.4 & 0.3 \\
sp-dt-cs & 32.5 & 2.3 & <0.05& 3.4 & 23.9 & 2.0 & - & 0.1 & - & - \\
sp-dt-csf & 4.8 & 2.1 & <0.05& 0.1 & 0.8 & 1.6 & - & <0.05& - & - \\
sp-dt-csv & 4.9 & 2.2 & <0.05& 0.1 & 0.8 & 1.5 & - & 0.1 & - & - \\
sp-dt-hc & 7.2 & 2.2 & <0.05& 0.3 & 3.3 & 1.2 & - & 0.1 & - & - \\
sp-dt-ks & 13.3 & 2.1 & <0.05& 0.7 & 8.0 & 2.1 & - & 0.1 & - & - \\
sp-dt-lf & 6.8 & 2.1 & <0.05& 0.2 & 2.8 & 1.4 & - & 0.1 & - & - \\
sp-dt-rb & 6.4 & 2.2 & <0.05& 0.2 & 2.3 & 1.4 & - & 0.1 & - & - \\
sp-dt-rc & 2.3 & 2.1 & <0.05& <0.05& <0.05& <0.05& - & <0.05& - & - \\
sp-os-bc & 6.4 & 2.3 & <0.05& 0.1 & 1.3 & 1.7 & - & 0.1 & - & - \\
sp-os-bk & 5.4 & 2.3 & <0.05& 0.1 & 1.3 & 1.4 & - & 0.1 & - & - \\
sp-os-cl & 101.0 & 2.3 & 0.4 & 1.8 & 22.9 & 61.7 & 9.5 & 0.2 & - & - \\
sp-os-cn & 25.3 & 2.4 & 0.1 & 1.7 & 11.5 & 7.2 & 0.7 & <0.05& 0.4 & 0.3 \\
sp-os-cs & 13.7 & 2.5 & <0.05& 1.2 & 7.4 & 1.8 & - & <0.05& 0.3 & 0.1 \\
sp-os-csf & 5.0 & 2.3 & <0.05& 0.1 & 0.8 & 1.6 & - & 0.1 & - & - \\
sp-os-csv & 5.1 & 2.3 & <0.05& 0.1 & 0.8 & 1.6 & - & 0.1 & - & - \\
sp-os-hc & 7.2 & 2.3 & <0.05& 0.3 & 3.1 & 1.2 & - & 0.1 & - & - \\
sp-os-ks & 12.9 & 2.3 & <0.05& 0.6 & 7.6 & 2.0 & - & 0.1 & - & - \\
sp-os-lf & 7.0 & 2.3 & <0.05& 0.2 & 2.7 & 1.5 & - & 0.1 & - & - \\
sp-os-rb & 6.7 & 2.4 & <0.05& 0.2 & 2.2 & 1.5 & - & 0.1 & - & - \\
sp-os-rc & 5.2 & 2.3 & <0.05& 0.2 & 1.0 & 1.4 & - & <0.05& - & - \\
sp-wt-bc & 5.5 & 2.3 & <0.05& 0.2 & 1.2 & 1.5 & - & 0.1 & - & - \\
sp-wt-bk & 5.2 & 2.3 & <0.05& 0.2 & 1.2 & 1.3 & - & 0.1 & - & - \\
sp-wt-cl & 85.3 & 2.3 & 0.3 & 1.5 & 18.8 & 52.6 & 6.7 & 0.1 & - & - \\
sp-wt-cn & 19.2 & 2.4 & 0.1 & 1.1 & 6.9 & 7.0 & 0.4 & <0.05& 0.4 & 0.3 \\
sp-wt-cs & 12.5 & 2.4 & <0.05& 1.1 & 6.6 & 1.7 & - & <0.05& 0.3 & 0.1 \\
sp-wt-csf & 5.0 & 2.3 & <0.05& 0.1 & 0.7 & 1.5 & - & 0.1 & - & - \\
sp-wt-csv & 5.0 & 2.3 & <0.05& 0.1 & 0.7 & 1.5 & - & 0.1 & - & - \\
sp-wt-hc & 6.9 & 2.3 & <0.05& 0.2 & 2.8 & 1.2 & - & 0.1 & - & - \\
sp-wt-ks & 11.7 & 2.3 & <0.05& 0.6 & 6.6 & 1.8 & - & 0.1 & - & - \\
sp-wt-lf & 6.5 & 2.3 & <0.05& 0.2 & 2.3 & 1.4 & - & 0.1 & - & - \\
sp-wt-rb & 6.2 & 2.3 & <0.05& 0.2 & 1.9 & 1.4 & - & 0.1 & - & - \\
sp-wt-rc & 5.0 & 2.3 & <0.05& 0.2 & 0.9 & 1.4 & - & 0.1 & - & - \\
abc32-cmp & 0.3 & 0.3 & <0.05& <0.05& <0.05& <0.05& - & <0.05& - & - \\
abc32-dc2 & 0.2 & 0.2 & <0.05& <0.05& <0.05& <0.05& - & <0.05& - & - \\
abc32-resyn & 0.2 & 0.2 & <0.05& <0.05& <0.05& <0.05& - & <0.05& - & - \\
abc32-resyn2 & 0.2 & 0.2 & <0.05& <0.05& <0.05& <0.05& - & <0.05& - & - \\
abc32-resyn3 & 0.2 & 0.2 & <0.05& <0.05& <0.05& <0.05& - & <0.05& - & - \\
abc64-cmp & 2.3 & 2.1 & <0.05& <0.05& <0.05& <0.05& - & <0.05& - & - \\
abc64-dc2 & 2.3 & 2.1 & <0.05& <0.05& <0.05& <0.05& - & <0.05& - & - \\
abc64-resyn & 2.3 & 2.1 & <0.05& <0.05& <0.05& <0.05& - & <0.05& - & - \\
abc64-resyn2 & 2.3 & 2.1 & <0.05& <0.05& <0.05& <0.05& - & <0.05& - & - \\
abc64-resyn3 & 2.3 & 2.2 & <0.05& <0.05& <0.05& <0.05& - & <0.05& - & - \\
abc128-cmp & 34.0 & 32.6 & <0.05& <0.05& <0.05& <0.05& - & 0.1 & - & - \\
abc128-dc2 & 33.4 & 32.0 & <0.05& <0.05& <0.05& <0.05& - & 0.1 & - & - \\
abc128-resyn & 33.5 & 32.0 & <0.05& <0.05& <0.05& <0.05& - & 0.1 & - & - \\
abc128-resyn2 & 34.5 & 33.0 & <0.05& <0.05& <0.05& <0.05& - & 0.1 & - & - \\
abc128-resyn3 & 33.8 & 32.4 & <0.05& <0.05& <0.05& <0.05& - & 0.1 & - & - 

\end{longtable}
}

\end{document}